\documentclass[10pt]{article}

\usepackage{bm}
\usepackage{fullpage}  % for narrow margins

\usepackage{url}
\usepackage{algorithm, algorithmicx,algpseudocode}
\usepackage{multirow} %//multirow for format of table
\usepackage{hhline}  % for \hhline in tables

\usepackage{graphicx}
\usepackage{amsmath,amssymb,amsfonts,graphicx,amsthm,mathtools,nicefrac,mathrsfs}
\usepackage{lscape}
\usepackage{color}
\usepackage{authblk}
\usepackage{subfigure}

\allowdisplaybreaks
\newcommand\numberthis{\addtocounter{equation}{1}\tag{\theequation}}
\newtheorem{definition}{Definition}[section]
\newtheorem{theorem}{Theorem}[section]
\newtheorem{lemma}{Lemma}[section]

\newtheorem{remark}{Remark}[section]

\numberwithin{equation}{section}

\DeclareMathOperator{\rank}{\mathrm{rank}}
\DeclareMathOperator{\diag}{\mathrm{diag}}

\DeclareMathOperator{\trace}{trace}

\def\la{\langle}
\def\ra{\rangle}

\newcommand{\vecnorm}[2]{\left\| #1\right\|_{#2}}

\newcommand{\matsnorm}[2]{\left\| #1\right\|_{{#2}}}

\newcommand{\nucnorm}[1]{\ensuremath{\matsnorm{#1}{\footnotesize{\mbox{$\ast$}}}}}
\newcommand{\fronorm}[1]{\ensuremath{\matsnorm{#1}{\footnotesize{\mathsf{F}}}}}
\newcommand{\opnorm}[1]{\ensuremath{\matsnorm{#1}{}}}

\newcommand{\bfm}[1]{\ensuremath{\mathbf{#1}}}

\renewcommand{\Pr}[2][]{\mathbb{P}_{#1} \left\{ #2 \rule{0mm}{3mm}\right\}}

\newcommand{\E}[2][]{\mathbb{E}_{#1} \left\{ #2 \rule{0mm}{3mm}\right\}}

   \def\mA{\bfm A}  
   \def\mB{\bfm B}  
     \def\C{\mathbb{C}}
   \def\mD{\bfm D}  
\def\ve{\bfm e}   \def\mE{\bfm E}  
  \def\mF{\bfm F}  
   \def\mG{\bfm G}

   \def\mL{\bfm L}

   \def\mR{\bfm R}  \def\R{\mathbb{R}}
\def\vs{\bfm s}   \def\mS{\bfm S}  
     
   \def\mU{\bfm U}  
   \def\mV{\bfm V}  
   \def\mW{\bfm W}  
\def\vx{\bfm x}   \def\mX{\bfm X}  
\def\vy{\bfm y}   \def\mY{\bfm Y}  
\def\vz{\bfm z}   \def\mZ{\bfm Z}

\def\calA{{\cal  A}}

\def\calD{{\cal  D}}

\def\calG{{\cal  G}} 
\def\calH{{\cal  H}} 
\def\calI{{\cal  I}}

\def\calO{{\cal  O}} 
\def\calP{{\cal  P}}

\def\calW{{\cal  W}}

\def\calZ{{\cal  Z}} 
\def\bzero{\bfm 0}

\newcommand{\bfsym}[1]{\ensuremath{\boldsymbol{#1}}}

             \def\bSigma{\bfsym \Sigma}
        \def\bLambda {\bfsym {\Lambda}}

\def \tE{\widehat{\mE}}

\def \tU{\widehat{\mU}}
\def \tV{\widehat{\mV}}
\def \tW{\widehat{\mW}}
\def \tX{\widehat{\mX}}

\def \tZ{\widehat{\mZ}}
\def \tA{\widehat{\mA}}
\def \tB{\widehat{\mB}}

\def \tcalG{\widehat{\calG}}
\def \tcalGT{\widehat{\calG}^{\ast}}

\def \tT{\widehat{T}}
\def \tmG{\widehat{\mG}}
\def \tSigma {\widehat{\bSigma}}

\def \tran {\mathsf{T}}
\def \tranH{\mathsf{H}}

\def \bzero{\bm 0}
\def \bone{\bm 1}

\newcommand{\kw}[1]{{{#1}}}
\DeclareMathOperator*{\minimize}{minimize}
\begin{document}
	
\title{Exact matrix completion based on low rank Hankel structure in  the Fourier domain}
\author[1]{Jinchi Chen}
\author[1,2]{Weiguo Gao}
\author[1]{Ke Wei}
\affil[1]{School of Data Science, Fudan University, Shanghai, China.\vspace{.15cm}}
\affil[2]{School of Mathematical Science, Fudan University, Shanghai, China.}
%\affil[3]{Oden Institute of Computational Engineering and Sciences, The University of Texas at Austin, Austin, Texas, USA.}

\date{\today}
%%%%%%
\maketitle
%%%%%%
\begin{abstract}
Matrix completion is about recovering a matrix from its partial revealed entries, and it can often be achieved by exploiting the  inherent simplicity or low dimensional structure of the target matrix. For instance, a typical notion of matrix simplicity is low rank. In this paper we study matrix completion based on another low dimensional structure, namely the low rank Hankel structure in the Fourier domain. It is shown that matrices with this structure can be exactly recovered by solving a convex  optimization program provided the sampling complexity is nearly optimal. Empirical results are also presented to justify the effectiveness of the convex method.
\end{abstract}
%%%%%%%%%%%%%%%%
%%%%%%%%%%%%%%%%
\section{Introduction}\label{sec:introduction}
This paper considers the problem of matrix completion which is about filling in the missing entries of a matrix. Letting $\mX^\natural\in\C^{d\times n}$ be the target matrix and  $\Omega\subset[d]\times [n]$ be a subset of indices corresponding \kw{to} the observed entries, the matrix completion problem can be expressed as: 
\begin{align*}%\it\centering
\mbox{\it Find the matrix }\mX^\natural~\mbox{\it from }\calP_\Omega(\mX^\natural),
\end{align*}
where $\calP_\Omega$ denotes the sampling operator which only acquires the matrix entries in $\Omega$. Despite the simplicity of this problem, it has many applications, such as sensor network localization \cite{so2007theory}, collaborative filtering \cite{rennie2005fast}, and multi-class learning \cite{argyriou2008convex}.

Without any additional assumptions, matrix completion is an ill-posed problem which does not even have a unique solution. Therefore computationally efficient solution to this problem is typically based on certain intrinsic low dimensional structures of the matrix. A notable example is low rank matrix completion where the target matrix is \kw{assumed} to be low rank. From the pioneering work of Cand\`es and Recht \cite{candes2009exact}, low rank matrix completion has received plenty of attention both  from the theoretical and algorithmic aspects, see \cite{recht2011simpler,gross2011recovering,chen2015incoherence,davenport2016overview,cai2018exploiting,chen2018harn,chi2019nonconvex} and references therein.

In this paper a different simple structure in the Fourier domain will be exploited for matrix completion. Let $\calH$ be a linear operator which maps a vector $\vx\in\C^{1\times n}$ into an $n_1\times n_2$ Hankel matrix, 
\begin{align*}
\calH(\vx) = \begin{bmatrix}
\vx_1  & \vx_2 & \vx_3&\cdots &\vx_{n_2}\\
\vx_2  & \vx_3 & \vx_4&\cdots &\vx_{n_2+1}\\
\vx_3  & \vx_4 & \vx_5&\cdots &\vx_{n_2+2}\\
\vdots & \vdots& \vdots& \ddots&\vdots\\
\vx_{n_1}&\vx_{n_1+1}&\vx_{n_1+2}&\cdots &\vx_{n}\\
\end{bmatrix}\in\C^{n_1\times n_2},
\end{align*}
where \kw{$\vx_i$ is the $i$th entry of $\vx$} and $n+1=n_1+n_2$ . Let $\tX^\natural\in\C^{d\times n}$ be \kw{the} matrix obtained by applying the Discrete Fourier Transform (DFT) to each column of $\mX^\natural$; namely, $\tX^\natural=\mF\mX^\natural$, where  $\mF\in \C^{d\times d}$ is the unitary DFT matrix. We will study matrix completion based on the low rank structure of the Hankel matrices associated with each row of $\tX^\natural$, which corresponds to different frequencies in the Fourier domain; see Figure~\ref{fig:illustration} for an illustration. Specifically, assuming
\begin{align}
\label{def: hankel low rank structures}
\rank(\calH(\tX^\natural(i,:))) = r\ll \min(d,n), \text{ for all } ~i=1,\cdots,d,
\end{align}
we ask when and how we are able to \kw{recover}
%recovery
the missing entries of $\mX^\natural$ from the observed ones. The low rank Hankel structure in the Fourier domain also arises in many applications,  including seismic data de-noising
and reconstruction \cite{oropeza2010singular,oropeza2011,huang2016damped} and the finite rate of innovation model \cite{vetterli2002sampling,blu2008sparse}.

\begin{figure}[t]
	\centering
	\subfigure{
		\begin{minipage}[t]{0.22\linewidth}
			\centering
			\centerline{\includegraphics[width=1\textwidth]{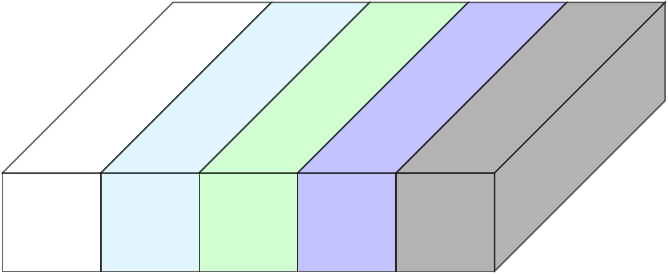}}
			\vspace{0.3cm}
			\centerline{\hspace{-.8cm}$\mX^\natural$}
		\end{minipage}%
	}%
	\subfigure{
		\begin{minipage}[t]{0.22\linewidth}
			\centering
			\centerline{\includegraphics[width=1\textwidth]{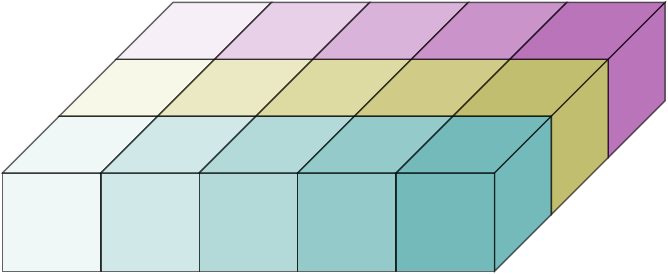}}
			\vspace{0.3cm}
			\centerline{\hspace{-.8cm}$\tX^\natural=\mF\mX^\natural$}
		\end{minipage}%
	}%
	\subfigure{
		\begin{minipage}[t]{0.35\linewidth}
			\centering
			\centerline{\includegraphics[width=0.7\textwidth]{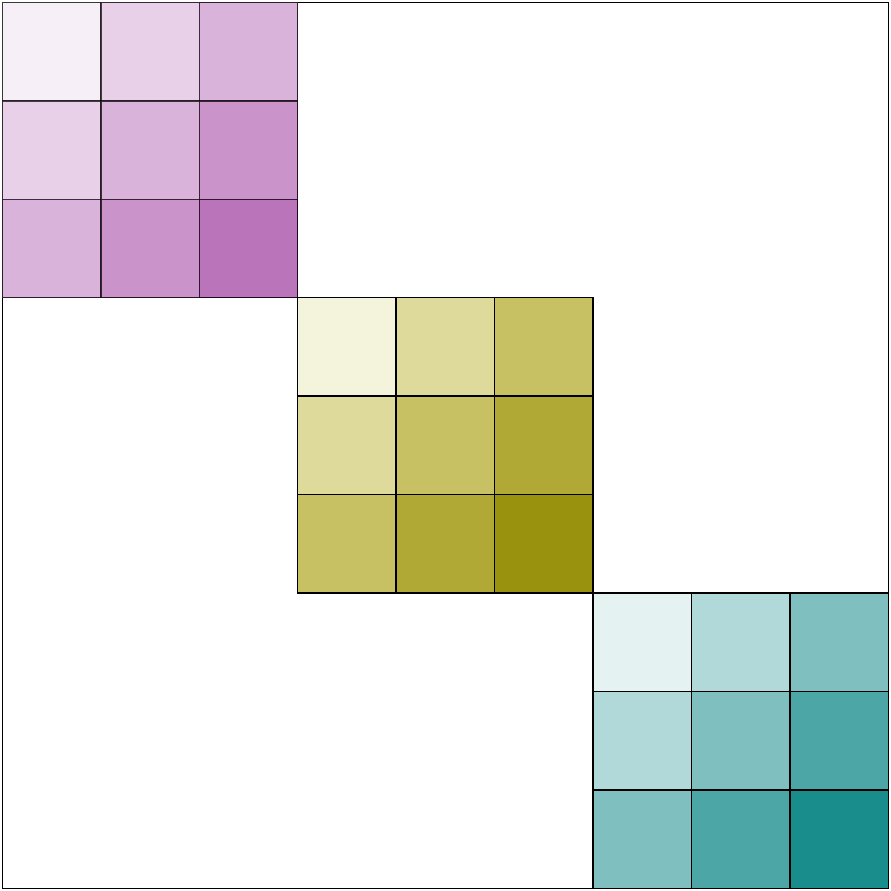}}
			\vspace{0.3cm}
			\centerline{$\diag(\calH(\tX^\natural(1,:)),\cdots,\calH(\tX^\natural(d,:)))$}
		\end{minipage}
	}%
	
	\centering
	\caption{ Illustration of  the data structure model.}\label{fig:illustration}
\end{figure}

\paragraph{Motivating example from seismic data analysis}
%% check this example
In seismology, a signal can be modeled as a superposition of plane waves in the temporal-spatial domain. For simplicity, % One motivating application for the structured matrices considered herein is in the seismic data model. In such case, 
%Our signal model is inspired by real applications. For instance, in the seismic recovery problem,
%the 2-dimensional seismic signal can be modeled as a composition of $r$ dipping events (plane waves):
consider a 2-dimensional seismic signal with $r$ dipping events (plane waves),
\begin{align}
\label{eq: 2D seismic}
s(x, t) = \sum_{\ell=1}^{r} u_\ell(t-p_\ell\cdot  x),
%s(x, t) = \sum_{\ell=1}^{r} d_{\ell} \cdot \delta(t - p_\ell \cdot x ),
\end{align}
where $t$ denotes the temporal direction, $x$ denotes the spatial direction, $u_\ell(\cdot)$ is a pulse or wavelet function, %such that if $t = p_\ell x$, then $\delta(t - p_\ell x )=1$, otherwise $\delta(t - p_\ell x )=0$, 
 and $p_\ell$ is the ray parameter of each waveform. After applying the Fourier transform to each spatial trace, we have
\begin{align}
\hat{s}(x,\omega) = \sum_{\ell=1}^{r}A_\ell(\omega)e^{-i\omega p_\ell x}, 
\end{align}
where $\omega$ denotes temporal frequency, and $A_\ell(\omega)$ represents the amplitude for the $\ell$th event given by
\begin{align*}
A_\ell(\omega) = \int_{-\infty}^{\infty} u_\ell(t)e^{-i\omega t}dt.
\end{align*}

By replacing the spatial variable $x$ with its discrete counterpart $x=k\Delta_x$ for $k=1,\cdots, n$, we obtain  the discrete signal for one monochromatic temporal frequency $\omega$,
\begin{align}
\hat{ \vs} = \begin{bmatrix}
\hat{\vs}_1,\cdots, \hat{\vs}_{n}
\end{bmatrix},~\mbox{where }\hat{s}_k:=\hat{ s}(k\Delta_x, \omega) = \sum_{\ell=1}^{r}A_\ell(\omega)e^{-i\omega p_\ell k\Delta_x}.
\end{align}
Letting $y_\ell = e^{-i\omega p_\ell\Delta_x}$,  then a simple algebra yields that the Hankel matrix $\calH\hat{ \vs}$ corresponding to $\hat{ \vs}$ has the following Vandermonde decomposition:
\begin{align*}
\calH(\hat{ \vs})= \mL\mD\mR^\tran,
\end{align*}
where the matrices $\mL$ and $\mR$ are given by
\begin{align*}
\mL= \begin{bmatrix}
1 	& 1 	&\cdots & 1\\
y_1	& y_2 	&\cdots &y_r\\
\vdots& \vdots &\ddots & \vdots\\
y_1^{n_1-1} & y_2^{n_1-1} &\cdots &y_r^{n_1-1}\\
\end{bmatrix},~\mR= \begin{bmatrix}
1 	& 1 	&\cdots & 1\\
y_1	& y_2 	&\cdots &y_r\\
\vdots& \vdots &\ddots & \vdots\\
y_1^{n_2-1} & y_2^{n_2-1} &\cdots &y_r^{n_2-1}\\
\end{bmatrix},
\end{align*}
 and $\mD = \diag(A_1(\omega),\cdots, A_r(\omega))$ is a diagonal matrix. If all $y_\ell$'s are distinct and all $A_\ell(\omega)$'s are non-zeros,  it is not hard to see that $\calH(\hat{ \vs})$ is a low rank matrix when $r\ll \min(n_1,n_2)$. Therefore, the seismic signal exhibits a low rank Hankel structure in the Fourier domain for each fixed frequency. In fact, this structure has been  extensively  used in seismic data processing, see for example \cite{oropeza2009multifrequency,sacchi2009fx,trickett2008f,oropeza2011simultaneous,trickett2009prestack,wang2019fast}.

%%%%%%%
\subsection{Methodology}
In low rank matrix completion, one of the most widely studied method is nuclear norm minimization, where the nuclear norm of a matrix is the sum of its singular values. As the tightest convex envelope of the matrix rank, minimizing the nuclear norm of a matrix is able to promote the low rank structure.
Here we also attempt to complete $\mX^\natural$ by solving a nuclear norm minimization problem,
\begin{align}
\label{opt: our nuclear norm}
\minimize_{\mX\in\C^{d\times n}}~\sum_{i=1}^{d} \nucnorm{\calH\left(\ve_i^\tran \mF\mX\right)}~\text{subject to }~\calP_{\Omega}(\mX) = \calP_{\Omega}(\mX^\natural),
\end{align}
where $\ve_i$ is the $i$-th standard orthonormal basis of $\R^d$, and  $\ve_i^\tran \mF\mX$ is just the $i$-th row of $\mF\mX$.

Note that \eqref{opt: our nuclear norm} is a convex program which can be solved by some well established software packages. Therefore we restrict our attention to the theoretical side of the problem and study when the solution to \eqref{opt: our nuclear norm} \kw{coincides} with $\mX^\natural$. In a nutshell, our result shows that 
\begin{quote}\centering\em
$\mX^\natural$ can be reconstructed  from about $\calO(dr\log^3 (dn))$ observed entries by solving \eqref{opt: our nuclear norm}.
\end{quote}
When $\mX^\natural$ satisfies \eqref{def: hankel low rank structures}, it has $\calO(dr)$ degrees of freedom which can be counted in the Fourier domain. Thus the sampling complexity established in this paper is optimal up to a logarithm factor. The details of the main theoretical result will be presented in Section~\ref{sec:main result}. Next we consider the recovery of a special matrix.
%%%%%%%
\paragraph{A special example}\label{para: special example}
Let $\mX^\natural$ be a rank-1 matrix given by
\begin{align}\label{eq:specialX}
\mX^\natural = \begin{bmatrix}
1		&0		&\cdots &0\\
1		&0		&\cdots &0\\
\vdots 	&\vdots &\ddots &\vdots \\
1		&0		&\cdots &0\\
\end{bmatrix} = \bone \ve_1^\tran\in\R^{d\times n}.
\end{align}
Since the right singular vector of $\mX^\natural$ is aligned with the canonical basis, we cannot expect to recover $\mX^\natural$ from $\calP_\Omega(\mX^\natural)$ by solving the standard nuclear norm minimization problem,
\begin{align}
\label{opt: nuclear norm}
\kw{\minimize_{\mX\in\C^{d\times n}} ~\nucnorm{\mX}}
%\min_{\mX\in\C^{d\times n}}\quad \sum_{i=1}^{d} \nucnorm{\mX},
~\text{subject to}~\calP_{\Omega}(\mX) = \calP_{\Omega}(\mX^\natural).
\end{align}
Indeed, one can easily show that the solution to \eqref{opt: nuclear norm} cannot be $\mX^\natural$ unless all the 1's in the first column of $\mX^\natural$ are observed.

\begin{figure}[t!]
	\centering
	\includegraphics[width=0.45\linewidth]{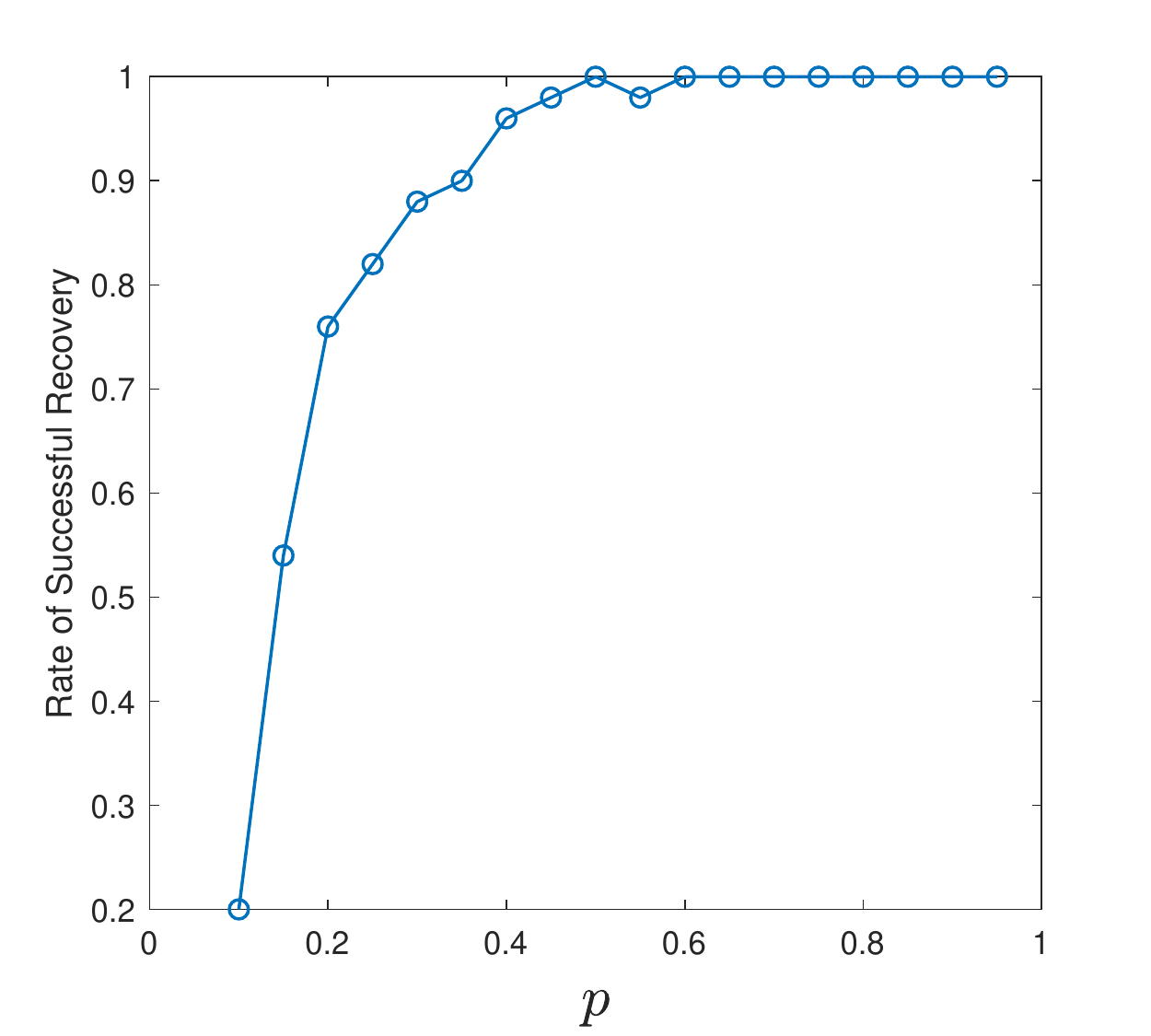}
	\caption{Rate of successful recovery {\it v.s.} Sampling probability for the special example.}
	\label{fig:special}
\end{figure}

To see whether \eqref{opt: our nuclear norm} can be used to fill in the missing entries of $\mX^\natural$, we first test this by Monte Carlo simulation. Here \eqref{opt: our nuclear norm} is solved by SDPT3 \cite{tutuncu2001sdpt3} based on CVX \cite{cvx,gb08}.
Assume that each entry of $\mX^\natural$ is observed independently with probability $p$. In our simulation we set $d=16$, $n=47$, and test $18$ equispaced values of $p$ from $0.1$ to $0.95$. For each value of $p$, $50$ trials are tested and a trial is declared to be successful if the relative reconstruction error is less than $10^{-3}$.
The rate of successful recovery out of $50$ trials against the sampling probability is 
presented in Figure~\ref{fig:special}. As can be seen from the figure, $\mX^\natural$ can be successfully recovered with probability at least 0.9 when $p\geq 0.35$. This is essentially because after the Fourier transform the first column of $\mX^\natural$ becomes a $1$-sparse vector and the convex program \eqref{opt: our nuclear norm}
reduces to the $\ell_1$-minimization method. The following theorem provides a formal justification of this observation. 

\begin{theorem}
	\label{main results: special}
	 Suppose  $d=2^L$ for a positive integer $L$, and assume that each entry of the special matrix $\mX^\natural$  is observed independently with probability $p$. Then $\mX^\natural$ can be exactly recovered by \eqref{opt: our nuclear norm} with probability at least $\left(1-(1-p)^d - dp(1-p)^{(d-1)}\right)/2$, which approaches $0.5$ when $d$ is sufficiently large.
\end{theorem}
We include the proof of this theorem  in Appendix~\ref{sec:special} since the key idea behind the proof is actually the same as that for the proof of the main result. Overall,  a dual variable needs to be constructed to certify the optimality of $\mX^\natural$.  In addition,  we would like to emphasize that the special example here merely shows a different way to exploit the matrix structure, and it by no means suggests that the low rank Hankel structure in the Fourier domain is more suitable than the low rank structure of the matrix itself for all the matrix completion problems.

%%%%%%%%%%%
\subsection{Related work}
In addition to low rank matrix completion, this work is also related to spectrally compressed sensing, low rank matrices demixing, and tensor completion via t-SVD.
\paragraph{Spectrally compressed sensing} 
This line of research is about reconstructing a spectrally sparse signal from a small number of time domain samples. It is indeed a special case of the problem studied in this paper when $d=1$. Both convex methods \cite{chen2014robust,bhaskar2013atomic,tang2013compressed} and nonconvex methods \cite{cai2019fast,cai2018spectral} have been proposed and analyzed for spectral compressed sensing. In particular, a method  called Enhanced Matrix Completion (EMaC) is studied  in the work of Chen and Chi \cite{chen2014robust}, which seeks a successful recovery by minimizing the nuclear norm of the Hankel matrix associated with the  signal.

\paragraph{Low rank matrices demixing} Demixing problems appear in many applications and low rank matrices demixing is concerned with the recovery of a few low rank matrices simultaneously 
from a single measurement vector. In \cite{mccoy2013achievable}, McCoy and Tropp consider the demixing model of the form $\vy = \mA \left( \sum_{i=1}^d  \mU_i \mX_k \right)$ and they propose to
reconstruct $\{\mX_i\}$ by minimizing a sum of the nuclear norms of all the constituent matrices.
Computationally efficient methods have \kw{been} developed for low rank matrices demxing in \cite{strohmer2019painless} and the exact recovery guarantees have been established based on the Gaussian measurement model. A joint blind deconvolution and blind demixing problem has been studied in \cite{ling2017blind,jung2017blind}, which can be reformulated as a demixing problem of rank-1 matrices from rank-1 measurements. As we will see later, the problem of recovering $\mX^\natural$ from $\calP_\Omega(\mX^\natural)$ based on \eqref{def: hankel low rank structures} can also be reformulated as a low rank matrices demixing problem.

\paragraph{Tensor completion via t-SVD} A tensor is an $N$-way array. When $N\geq 3$, there are various ways to define tensor ranks, typically based on different tensor factorizations such as the
 CANDECOMP/PARAFAC (CP)  factorization \cite{carroll1970analysis,harshman1970foundations,kolda2009tensor}, the Tucker factorization \cite{tucker1966some,grasedyck2010hierarchical}, the Tensor Train factorization \cite{oseledets2011tensor} , and the t-SVD factorization \cite{braman2010third,kilmer2011factorization,kilmer2013third,gleich2013power}. The philosophy behind our work is similar to that behind the tensor completion model via t-SVD \cite{zhang2017tensor,lu2018exact,lu2019tensor}, both of which utilize the low rank structure in the transform domain for data reconstruction. That being said, the t-SVD model cannot be used for matrix completion but  only work for high dimensional tensors as it does not involve another transformation from vectors to structure matrices. Therefore, our result cannot be covered by those for tensor completion via t-SVD.
 
 %%%%%%%%%
 \subsection{Organization}
 The remainder of this paper is organized as follows. Notation and preliminaries that are useful for our analysis are given in Section \ref{sec: preliminary}. The exact recovery guarantee of \eqref{opt: our nuclear norm} and numerical simulations are presented in Section \ref{sec:main result}. The proofs of the exact recovery guarantee  are provided from Section \ref{sec:proof1} to Section \ref{proof lemma key02to05}. %We conclude this paper in Section \ref{sec: conclusions} with a few future directions. 
\section{Notation and Preliminaries}
\label{sec: preliminary}
Throughout this paper vectors and matrices are denoted by bold lowercase and bold uppercase letters, respectively.  Moreover,  we often consider matrices in the Fourier domain, so there will be a hat over the letter in this case. Operators are denoted by calligraphic letters. In particular,  $\calI$ denotes the identity operator. 
For any matrix $\mX$, $\opnorm{\mX}, \fronorm{\mX}$ and $ \nucnorm{\mX}$ denote its spectral norm, Frobenius norm and nuclear norm, respectively. Given two complex matrices $\mX$ and $\mY$, their inner product is given by $\la\mX, \mY\ra = \trace(\mX\mY^\tranH)$. For a natural number $d$, we use $[d]$ to denote the set $\{1,\cdots, d\}$. The conjugate of a complex number $x\in\C$ is denoted by $\bar{x}$.% \kw{We use  $c,c_1,c_2,\cdots$ to denote the positive constants whose values may change from line to line.
%}

Recall that $\calH$ is a linear operator which associates a vector with a Hankel matrix. The adjoint of $\calH$, denoted $\calH^*$, is a linear mapping from $n_1\times n_2$ matrices to vectors of length $n$, 
\begin{align*}
[\calH^\ast(\mX)](a) = \sum_{j+k=a+1}\mX_{j,k},\quad\text{ for any }\mX\in\C^{n_1\times n_2},1\leq a\leq n.
\end{align*}
Let $\calD^2=\calH^*\calH$. One can easily verify that $\calD^2$ maps any vector $\vx\in\C^n$ to $\calD^2\vx = \{w_a\vx_a\}_{a=1}^n$, where $w_a$ denotes the number of entries on the $a$-th anti-diagonal of an $n_1\times n_2$ matrix, 
\begin{align*}
w_a = \#\left\{ (j,k)~|~j+k=a+1 ,1\leq j\leq n_1, 1\leq k\leq n_2 \right\}. \numberthis\label{eq:wi}
\end{align*}
Define $\calG=\calH\calD^{-1}$. Then its adjoint is given by $\calG^*=\calD^{-1}\calH$ and moreover we have $\calG^*\calG=\calI$. In addition, 
$\{\mG_a = \calG(\ve_a)\}_{a=1}^n$ form an orthogonal basis of the set of Hankel matrices, where $\ve_a$ is the $a$-th standard basis vector of $\R^n$.

With $\calG$ defined as above, let $\tcalG$ be an operator which maps a $d\times n$ matrix $\mX$ to  a $d$-block diagonal matrix of the form
\begin{align}
\label{def: G hat}
\tcalG(\mX) = \begin{bmatrix}
\calG(\ve_1^\tran\mF\mX)&&\\
&\ddots&\\
&&\calG(\ve_d^\tran\mF\mX)\\
\end{bmatrix}\in\C^{dn_1\times dn_2}.
\end{align}
The adjoint of $\tcalG$, \kw{denoted $\tcalGT$}, is given by
\begin{align}
\label{def G hat adjoint}
\tcalG^\ast(\tZ) = \mF^{-1}\left(\sum_{i=1}^{d}\ve_i\calG^\ast(\tZ_i) \right)\in\C^{d\times n},\text{ where $\tZ\in\C^{dn_1\times dn_2}$ is a $d$-block diagonal matrix}.
\end{align}
Furthermore, letting 
\begin{align}
\label{eq: Gjk}
\tcalG(\ve_j\ve_k^\tran)=\tmG_{j,k}=
\begin{bmatrix}
\ve_1^\tran\mF\ve_j\mG_k&&\\
&\ddots&\\
&&\ve_d^\tran\mF\ve_j\mG_k\\
\end{bmatrix}
\in\C^{dn_1\times dn_2},
\end{align}
we have the following lemma.
\begin{lemma}\label{lem:property 2}
The spectral norm of $\tmG_{j,k}$ satisfies $\|\tmG_{j,k}\|\leq{1}/{\sqrt{d w_k}}.$
\end{lemma}
\begin{proof}
It follows directly from the facts
$\|\mG_k\|\leq 1/\sqrt{w_k}$ and $\|\tmG_{j,k}\|=\max_{1\leq i\leq d}\|(\ve_i^\tran\mF\ve_j)\mG_k\|$.
\end{proof}

Let $\tZ_i= \tU_i\tSigma_i\tV_i^\tranH$ be the compact singular value decomposition (SVD) of a rank-$r$ matrix, where $\tU_i\in\C^{n_1\times r}, \tV_i\in\C^{n_2\times r}$ and $\tSigma_i\in\R^{r\times r}$. The sub-differential of $\|\cdot\|_*$ at \kw{$\tZ_i$} is given by  \cite{watson1992characterization}
\begin{align}
\partial \nucnorm{\tZ_i} = \tU_i\tV_i^\tranH + \left\{ \tW_i: \tW_i^\tranH\tU_i =\bzero, \tW_i\tV_i = \bzero, \opnorm{\tW_i}\leq 1 \right\}.
\end{align}
It is also known that the tangent space of the fixed rank-$r$ matrix manifold at \kw{$\tZ_i$} is given by \cite{vandereycken2013low}
\begin{align*}
\tT_i = \left\{ \tU_i\tA_i^\tranH + \tB_i\tV_i^\tranH: \tA_i\in\C^{n_2\times r},\tB_i\in\C^{n_1\times r} \right\}.
\end{align*}
Given any matrix $\tW_i\in\C^{n_1\times n_2}$, the projection of $\tW_i$ onto $\tT_i$ can be computed  using the formula  
\begin{align*}
\calP_{\tT_i}(\tW_i) = \tU_i\tU_i^\tranH\tW_i + \tW_i\tV_i\tV_i^\tranH -  \tU_i\tU_i^\tranH\tW_i\tV_i\tV_i^\tranH.
\end{align*}
Let $\tZ=\diag(\tZ_1,\cdots,\tZ_d)$, where each $\tZ_i$ is a rank-$r$ matrix with the compact SVD
$\tZ_i= \tU_i\tSigma_i\tV_i^\tranH$. 
Define $\tT=\diag(\tT_1,\cdots,\tT_d)$. 
For any $d$-block diagonal matrix $\tW=\diag(\tW_1,\cdots,\tW_d)$, the projection of $\tW$ onto $\tT$ is given by
\begin{align*}
\calP_{\tT}(\tW) = \begin{bmatrix}
\calP_{\tT_1}(\tW_1)&&\\
&\ddots&\\
&&\calP_{\tT_d}(\tW_d)
\end{bmatrix}.
\end{align*} 

Lastly, our analysis relies on the Bernstein inequality, which is stated as follows.
\begin{lemma}[\cite{tropp2012,chen2014robust}]
	Suppose $\{ \mX_\ell \}_{\ell=1}^m$ are independent random matrices  of dimension $d_1\times d_2$ and satisfy $\E{\mX_{\ell}} = \bzero$ and $\opnorm{\mX_\ell}\leq B$. Define
	\begin{align*}
	\sigma^2= \max\left\{ \opnorm{\sum_{\ell=1}^m\E{\mX_\ell\mX_\ell^\tranH}},\opnorm{\sum_{\ell=1}^m\E{\mX_\ell^\tranH\mX_\ell} } \right\}.
	\end{align*}
	Then the event 
	\begin{align}
	\label{bernstein}
	\opnorm{\sum_{\ell=1}^m\mX_{\ell}} \leq c\left( \sqrt{\sigma^2\log(d_1+d_2)} + B\log(d_1+d_2) \right)
	\end{align}
	occurs  with probability at least $1-(d_1+d_2)^{-c_1}$, where $c,c_1>0$ are absolute constants.
\end{lemma}
\section{Main results}
\label{sec:main result}
%In this section we present the main results of this paper which show that, under certain random models, $\tZ^\natural$ can be completed provided the sampling complexity is nearly optimal.
To study the recovery guarantee of \eqref{opt: our nuclear norm}, we first reformulate it as a block diagonal low rank Hankel matrix recovery problem. Given a matrix $\mX\in\C^{d\times n}$, let $\tZ=\diag(\tZ_1,\cdots,\tZ_d)$ be a $d$-block diagonal matrix defined as 
\begin{align*}
\tZ=\begin{bmatrix}
\calH(\ve_1^\tran\mF\mX)&&\\
&\ddots&\\
&&\calH(\ve_d^\tran\mF\mX)\\
\end{bmatrix}= \begin{bmatrix}
\calG(\ve_1^\tran\mF\mX\mD)&&\\
&\ddots&\\
&&\calG(\ve_d^\tran\mF\mX\mD)\\
\end{bmatrix}\in\C^{dn_1\times dn_2},
\end{align*}
where $\mD$ is diagonal matrix with $\mD_{ii} = \sqrt{w_i}$. Then we have the following facts:
\begin{align*}
(\calI - \tcalG\tcalG^\ast)(\tZ) = \bzero,~\sum_{i=1}^{d} \nucnorm{\calH\left(\ve_i^\tran \mF\mX\right)} = \nucnorm{\tZ},~\mbox{and}~\kw{\calP_{\Omega}(\mX\mD) = \calP_{\Omega}\tcalGT(\tZ)}.
%\calP_{\Omega}\tcalG^\ast(\tZ) = \calP_{\Omega}\tcalG^\ast(\tZ^\natural).
\end{align*}
Therefore, \eqref{opt: our nuclear norm} is equivalent to the following optimization problem in the Fourier domain,
\begin{align}
\label{opti: nuclear norm}
\minimize_{\tZ\in\C^{dn_1\times dn_2}}&\nucnorm{\tZ}\notag\\
\text{subject to}~&(\calI - \tcalG\tcalG^\ast)(\tZ) = \bzero\\
&\calP_{\Omega}\tcalG^\ast(\tZ) = \calP_{\Omega}\tcalG^\ast(\tZ^\natural),\notag
\end{align}
which will be our primary focus in the later analysis.
\begin{remark}{\normalfont
As we have already mentioned earlier, \eqref{opti: nuclear norm} (or equivalently, \eqref{opt: our nuclear norm}) is indeed a low rank Hankel matrices demixing problem. This can be seen by defining $\calA_i(\tZ_i) = \calP_{\Omega}(\mF^{-1}\ve_i\calG^\ast(\tZ_i))$ and further noticing that
$
	\calP_{\Omega}\tcalG^\ast(\tZ)=\sum_{i=1}^{d}\calA_{i}(\tZ_i),
$
}
\end{remark}
\kw{
\begin{remark}\normalfont
	Noticing that the $(j,k)$th entry of $\tcalGT(\tZ)$ is given by
	\begin{align*}
	[\tcalGT(\tZ)]_{j,k} = \la\ve_j\ve_k^\tran,\tcalGT(\tZ)\ra = \sum_{i=1}^{d}\la\ve_i^\tran\mF\ve_j \mG_k,\tZ_i\ra=\la\tmG_{j,k},\tZ \ra,
	\end{align*}
	the last constraint in \eqref{opti: nuclear norm} can be reformulated as
	\begin{align*}
	\la\tmG_{j,k}, \tZ \ra = \la\tmG_{j,k}, \tZ^\natural \ra,\quad(j,k)\in\Omega.
	\end{align*}
	This means  \eqref{opti: nuclear norm} is also a low rank matrix recovery problem from a few coefficients in an orthogonal basis \cite{gross2011recovering}. However, we note that the performance guarantee establish in \cite{gross2011recovering} does not apply here because we also require the matrix to be Hankel and block diagonal.\end{remark}
}
%%%%%%%%%%
\subsection{Exact recovery guarantee}
As is typical in the analysis for low rank matrix completion, the recovery guarantee of \eqref{opti: nuclear norm} relies on a notion of incoherence, which is defined as follows. Without loss of generality, we assume $n_1=n_2=(n+1)/2$  in the sequel; that is, $\calH$ maps a vector of length $n$ to a square Hankel matrix.

\begin{definition}[Average Case Incoherence Property]
\label{def: incoherence}
Let $\tZ^\natural$ be a $d$-block diagonal matrix. Assume the compact SVD of its $a$th-block diagonal is given by $\tZ_a^\natural = \tU_a\tSigma_a\tV^\tranH_a$. The matrix $\tZ^\natural$ is said to obey the incoherence condition with parameter $\mu_0$ if 
\begin{align}
\label{def: incoherent}
\max_{1\leq i\leq (n+1)/2}\frac{1}{d} \sum_{a=1}^{d}  \vecnorm{ \ve_i^\tran\tU_a}{2}^2 \leq  \frac{\mu_0r}{n}\quad\mbox{and}\quad \max_{1\leq j\leq (n+1)/2}\frac{1}{d} \sum_{a=1}^{d}  \vecnorm{ \ve_j^\tran\tV_a}{2}^2 \leq  \frac{\mu_0r}{n},
\end{align}
\kw{where $\ve_i$ is the $i$-th standard basis vector of suitable length.}
%hold for all $a=1,\cdots,d$.
\end{definition}

When $d=1$,  the average case incoherence property reduces to the standard incoherence property that has been widely used in \cite{chen2014robust,cai2019fast} for low rank matrix completion.  In this case, it means that the energy of the singular vector matrices is evenly distributed across all the rows so that  the singular vectors are
weakly correlated with the canonical basis. In general, Definition~\ref{def: incoherence} means that the average of the $i$-th rows of all the left singular vector matrices $\tU_a$ (respectively, the right singular vector matrices $\tV_a$) is in the same order for every $i$. Moreover, if $\{\tU_a,\tV_a\}_{a=1}^d$ satisfies the {\em worst case incoherence property} 
\begin{align}
\label{def worst incoh}
\max_{\substack{1\leq i\leq (n+1)/2\\1\leq a\leq d}} \vecnorm{ \ve_i^\tran\tU_a}{2}^2 \leq  \frac{\mu_1r}{n}\quad\mbox{and}\quad \max_{\substack{1\leq j\leq (n+1)/2\\1\leq a\leq d}}  \vecnorm{ \ve_j^\tran\tV_a}{2}^2 \leq  \frac{\mu_1r}{n},
\end{align}
it follows immediately that the average case incoherence property holds with $\mu_0=\mu_1$. However, the reverse direction is not true. As an example, consider the special matrix \kw{$\mX^\natural$} in \eqref{eq:specialX}. Multiplying \kw{$\mX^\natural$} by $\mF$ from the left yields 
\begin{align*}
\mF\kw{\mX^\natural} = \begin{bmatrix}
\sqrt{d} & 0& \cdots &0\\
0 &0 & \cdots &0\\
\vdots & \vdots &\ddots &\vdots\\
0 &0 & \cdots &0\\
\end{bmatrix}.
\end{align*}
Consequently,
\kw{
\begin{align*}
\tZ^\natural = \begin{bmatrix}
\calH(\ve_1^\tran\mF\mX^\natural)&&\\
&\ddots&\\
&&\calH(\ve_d^\tran\mF\mX^\natural)\\
\end{bmatrix}=
%\calG(\mZ^\natural)=
\begin{bmatrix}
\sqrt{d}\ve_1\ve_1^\tran&&&\\
& \bm{0}&&\\
&&\ddots&&\\
&&&\bm{0}
\end{bmatrix}.
\end{align*}
}
Therefore  we have 
\begin{align*}
\max_{1\leq i\leq (n+1)/2}\frac{1}{d} \sum_{a=1}^{d}  \vecnorm{ \ve_i^\tran\tU_a}{2}^2 =\frac{1}{d}\quad\mbox{and}\quad \max_{\substack{1\leq i\leq (n+1)/2\\1\leq a\leq d}} \vecnorm{ \ve_i^\tran\tU_a}{2}^2=1.
\end{align*}
It follows that the worst case incoherence property cannot be satisfied unless $\mu_1=\calO(n)$, but the average incoherence property may hold with $\mu_0=\calO(1)$ if $d$ is proportional to $n$.

The following lemma is a direct consequence of the average incoherence property, and  the short proof is provided in   Appendix~\ref{proof lemma incoherence 1}.
\begin{lemma}
	\label{lemma incoherence 1}
	\kw{Let $\tZ^\natural$ be a $d$-block diagonal matrix with the compact SVD of its $a$-th diagonal block being given by $\tZ_a^\natural = \tU_a\tSigma_a\tV^\tranH_a$. Suppose $\tZ^\natural$ satisfies the average case incoherence condition with parameter $\mu_0$.}
	%Suppose $\tZ^\natural$ is a $d$-block diagonal matrix and the compact SVD of its $a$th-block diagonal is $\tZ_a^\natural = \tU_a\tSigma_a\tV^\tranH_a$. 
	Then
	\begin{align}
	\label{ineq 1}
	\max_{k} \frac{1}{d}\sum_{a=1}^{d}\fronorm{\mG_k^\tranH\tU_a}^2 \leq \frac{\mu_0 r}{n}\quad\text{and}\quad\max_{k} \frac{1}{d}\sum_{a=1}^{d}\fronorm{\mG_k\tV_a}^2 \leq \frac{\mu_0r}{n}.
	\end{align}
	Furthermore, let $\tT$ be the induced tangent space of $\tZ^\natural$. Then for %any $1\leq a\leq d$ and 
	any $1\leq j\leq d, 1\leq k\leq n$, we have
	\begin{align}
	\label{ineq 2}
	\fronorm{\left(\calP_{\tT}\tcalG\right)(\ve_j\ve_k^\tran)}^2 \leq  \frac{2\mu_0r}{n}.
	\end{align}
\end{lemma}

We are now in the position to state the main theoretical result, the proof of which occupies a large part of this paper.
\begin{theorem}
\label{main result}
Suppose $\tZ^\natural$ satisfies the average case incoherence property with parameter $\mu_0$ and the index set $\Omega$ obeys the Bernoulli model with parameter $p$ (i.e., each entry of \kw{$\mX^\natural$} is sampled independently with probability $p$). 
Then $\tZ^\natural$ is the unique optimal solution to \eqref{opti: nuclear norm} with high probability provided $p\gtrsim \frac{\mu_0r\log^3(dn)}{n}$.
\end{theorem}
\kw{
\begin{remark}\normalfont Here and in the sequel, the notation $\gtrsim$ means that its lefthand side is greater than an absolute positive constant  times the righthand side.
	By high probability, we mean with probability at least $1-c_1(dn)^{-c_2}$ for some absolute constants $c_1,c_2>0$.
\end{remark}
}
\begin{remark}\normalfont
Although Theorem~\ref{main result} has been established under the Bernoulli sampling model, the sampling complexity can be translated to other sampling models, such as the sampling with replacement model, with a change in the constant factor. For conciseness, here we assume that all $\tZ_i^\natural$ have the same rank. It is worth noting that the analysis can be easily extended to the setting where each $\tZ_i^\natural$ has a different rank. Since the average case incoherence property holds for the special matrix in \eqref{eq:specialX}, Theorem~\ref{main result} also justifies the successful completion of the matrix, as we have observed from the simulation.
\end{remark}

%%%%%%%%%%%%
\subsection{Extension to higher dimension}
Though we have mainly focused on the \kw{two-dimensional} (2D) matrix completion problem, the model and analysis are also applicable for higher dimensional array recovery problem. For ease of exposition, we give a brief discussion of the \kw{three-dimensional} (3D) case.

For a 3D array $\mX\in\C^{n\times s \times d}$, we denote by $\mX_{j,k,i}$ the $(j,k,i)$-th entry of $\mX$ and use  $\mX(j,:,:), \mX(:,k,:), \mX(:,:,i)$ to denote the $j$-th horizontal, $k$-th lateral and $i$-th frontal slices, respectively. For simplicity, the frontal slice $\mX(:,:,i)$ is also denoted by $\mX_{i}\in\C^{n\times s}$. The $(j,k)$-th tube of $\mX$ is given by $\mX(j,k,:)\in\C^{d}$. Let $\tX^\natural$ be \kw{an array}
%a matrix
 obtained by applying the DFT to each of its tubes, i.e., $\tX^\natural(j,k,:) = \mF\mX^\natural(j,k,:)$ for all $(j,k)\in[n]\times [s]$.

Let $(L_1, K_1)$ and $(L_2, K_2)$ be two pairs of positive numbers satisfying $L_1 + K_1 = n+1$ and $L_2 + K_2 = s+1$. For each frontal slice $\tX^\natural_{i}$ of $\tX^\natural$, we can associate a  two-level block Hankel matrix with it as follows:
\begin{align}
\label{eq high dim hankelization}
\calH\tX^\natural_{i} := \begin{bmatrix}
\calH\tX^\natural_{i}(0,:) & \calH\tX^\natural_{i}(1,:) &\cdots &\calH\tX^\natural_{i}(K_1-1,:)\\
\calH\tX^\natural_{i}(1,:) & \calH\tX^\natural_{i}(2,:) &\cdots &\calH\tX^\natural_{i}(K_1,:)\\
					   \vdots & \vdots						  &\ddots &\vdots\\
\calH\tX^\natural_{i}(L_1-1,:) & \calH\tX^\natural_{i}(L_1,:) &\cdots &\calH\tX^\natural_{i}(n-1,:)\\
\end{bmatrix}\in\C^{L_1L_2\times K_1K_2},
\end{align}
where \kw{$\calH\tX^\natural_{i}(j,:)\in\C^{L_2\times K_2}$ is also a  Hankel matrix corresponding to the $j$-th row of $\tX^\natural_i$}. 

Assuming 
\begin{align*}
\rank(\calH\tX^\natural_{i}) = r\ll \min(L_1L_2, K_1K_2)\text{ for all } i=1,\cdots, n,
\end{align*}
we may reconstruct \kw{$\mX^\natural$} from its partial revealed entries by solving the following convex  program,
\begin{align*}
\minimize_{\mX\in\C^{n\times s\times d}}~\sum_{i=1}^{d} \nucnorm{\calH\tX_{i}}~\text{subject to }\begin{cases}\calP_{\Omega}(\mX) = \calP_{\Omega}(\mX^\natural)\\
\tX(j,k,:) = \mF\mX(j,k,:) \mbox{ and }\tX_i = \tX(:,:,i).\numberthis\label{eq:convex_higher}
\end{cases}
\end{align*}
Letting $\tZ^\natural_i = \calH\tX^\natural_i$ for $i=1,\cdots,d$ and 
\begin{align*}
\tZ^\natural = \begin{bmatrix}
\tZ^\natural_1 &&\\
&\ddots&\\
&&\tZ^\natural_d\\
\end{bmatrix}\in\C^{dL_1L_2\times dK_1K_2},
\end{align*}
the following recovery guarantee can be established for \eqref{eq:convex_higher}. The proof details are overall similar to that for Theorem~\ref{main result}, and thus are omitted.
\begin{theorem}
	\label{main result 2}
	Suppose $\tZ^\natural\in\C^{dL_1L_2\times dK_1K_2}$ satisfies the \kw{incoherence}
	%incoherent 
	condition \eqref{def: incoherent} with parameter $\mu_0$ and the index set $\Omega\subset[n]\times [s]\times [d]$ obeys the Bernoulli model with parameter $p$. If 
	$
	p\gtrsim \frac{\mu_0r\log^4(dns)}{ns},
	$
	then $\tZ^\natural$ is the unique optimal solution to \eqref{opti: nuclear norm} with high probability.% \emph{i.e.}, $1-c_1(dns)^{c_2}$ for some absolute constants $c_1,c_2>0$.
\end{theorem}
%%%%%%%%%%
%%%%%%%%%%
\subsection{Numerical experiments}
\label{sec: numerical experiments}
In this section we empirically evaluate the recovery performance of \eqref{opt: our nuclear norm}. Simulations have been performed for both the 2D and 3D data, and the results are presented in Section \ref{sec: 2dim case} and  Section \ref{sec: 3dim}, respectively.
\subsubsection{Phase transitions for matrix completion}
\label{sec: 2dim case}
The experiment setup is the same as that for the special example in the introduction. Two different types of matrices are tested, one of which is constructed in the following two steps:
\begin{itemize}
	\item Generate a random matrix $\tX^\natural\in\C^{16\times 47}$, where each row of $\tX^\natural$ is a random spectrally sparse signal of length $47$ with $r\in[1,24]$ frequency components. Specifically, 
	\begin{align*}
	\tX^\natural(i,:) = \begin{bmatrix}
	x(0) & x(1) & \cdots & x(46)
	\end{bmatrix},\numberthis\label{2D_test_signal}
	\end{align*}
	where $x(t) = \sum_{k=1}^{r}d_k \exp(2\pi i f_k t)$ for $t=0,\cdots, 46$. Here, $f_k$ is uniformly sampled from $[0,1)$, and the complex weight is generated via $d_k=(1+10^{0.5c_k})e^{i\psi_k}$ with $\psi_k$ being uniformly sampled from $[0,2\pi)$  and $c_k$ being uniformly sampled from $[0,1]$.
	\item Apply the Inverse Discrete Fourier Transform (IDFT) to each column of $\tX^\natural$ to get $\mX^\natural$, which gives the test matrix.	
\end{itemize}  
Note that based on the Vandermonde decomposition of  $\tX^\natural(i,:)$, one can easily see that the matrix constructed in this way satisfies the low rank Hankel property \eqref{def: hankel low rank structures}. The phase transition diagram for this type of matrix is presented in the left plot of Figure \ref{fig:2D}, where the horizontal axis denotes the sampling ratio $p=m/(dn)$ and the vertical axis denotes the rank $r$. The color of each cell reflects the empirical success rate. In particular, white color means that for a fixed pair of $(p,r)$ all of the $50$ random matrices can be successfully recovered by \eqref{opt: our nuclear norm}, while black color indicates that \eqref{opt: our nuclear norm} fails to recover any of the $50$ matrices. The plot shows a nearly linear scaling between the sampling ratio and the largest rank that  \eqref{opt: our nuclear norm} can achieve a successful recovery with high probability.

\begin{figure}[ht!] \centering
	\subfigure{
		\label{fig a}
		\includegraphics[width=0.45\columnwidth]{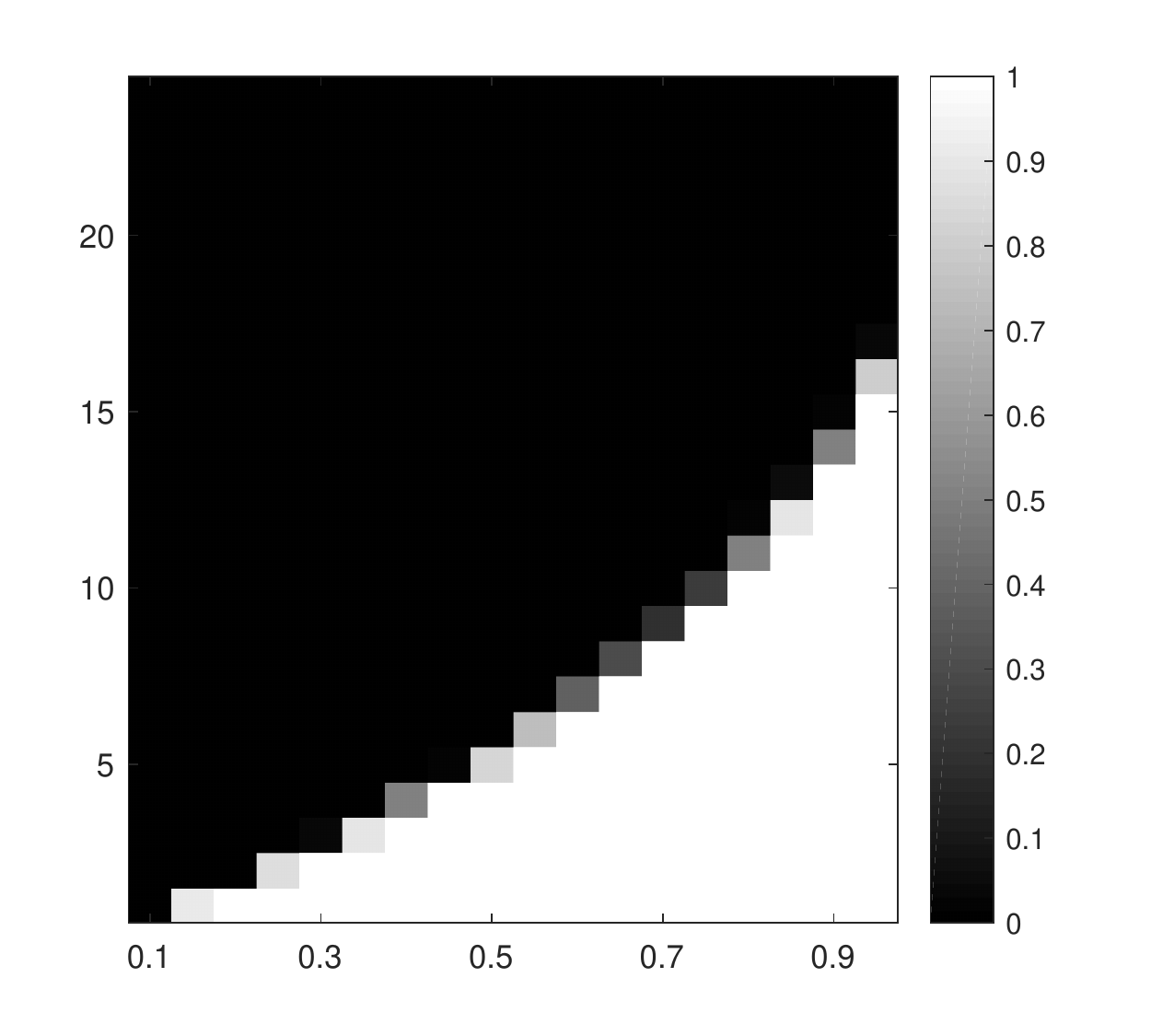}
	}
	\subfigure{
		\label{fig b}
		\includegraphics[width=0.45\columnwidth]{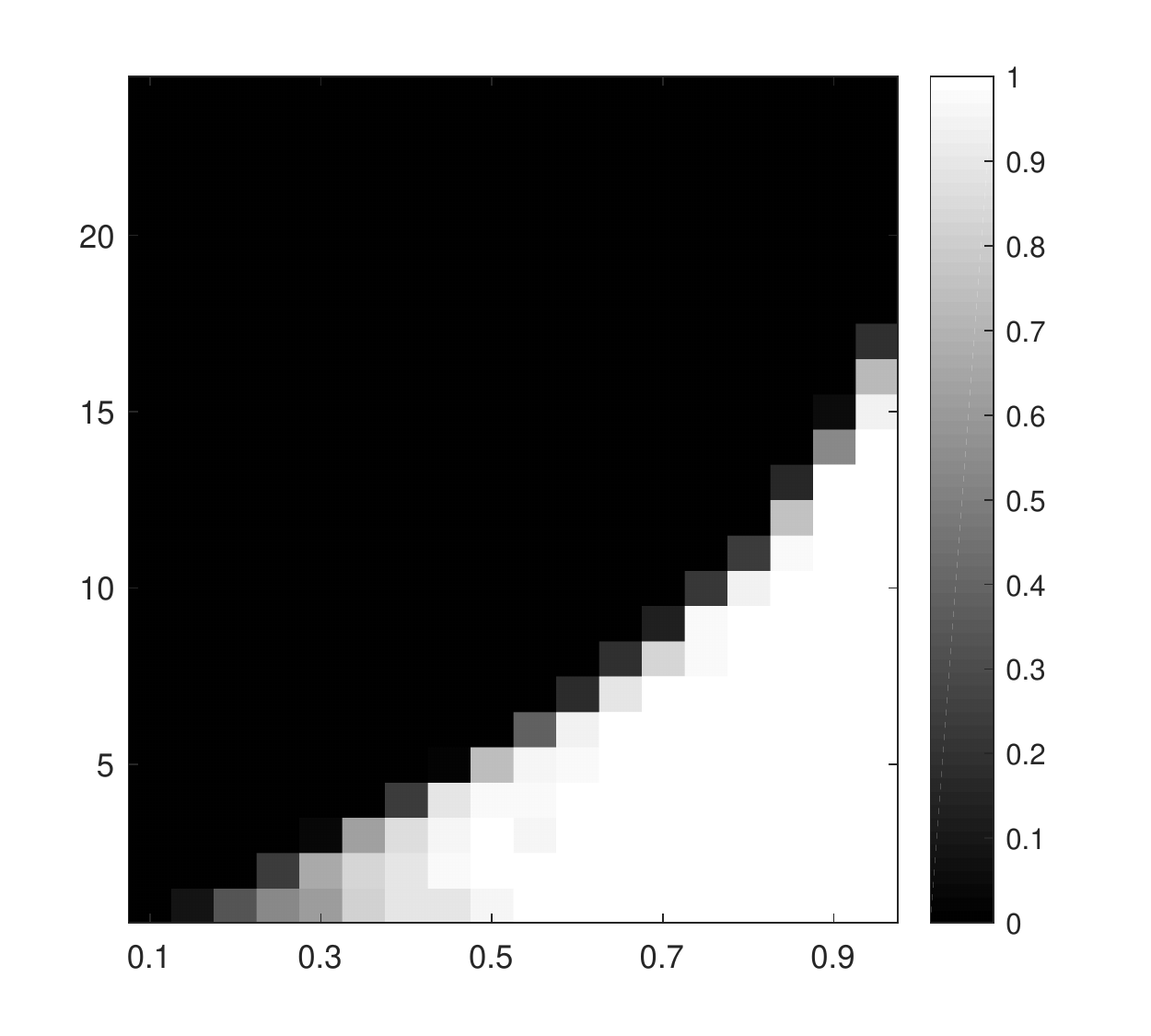}
	}
	
	\caption{Phase transition diagram for matrix completion. The horizontal axis is  $p=m/dn$ and the vertical axis is $r$. Left: test matrix obeys the worst case incoherent condition; Right: test matrix obeys the average case incoherent condition but does not obey the worst case incoherence condition.}
	\label{fig:2D}
\end{figure}

The test matrix constructed as above satisfies the worst case incoherence condition \eqref{def worst incoh} \cite{liao2016music}, and thus also satisfies the average case incoherence condition \eqref{def: incoherent}. To further examine the average case incoherence assumption in Theorem~\ref{main result}, we randomly replace two rows of $\tX^\natural$ in \eqref{2D_test_signal} with a random vector  $\vx\in\R^{47}$ generated in the following way. The entries of $\vx$ are zeros everywhere except the ones in positions $1 \leq j,k\leq 47$. The position $k$ is chosen uniformly at random  from the set $\{r, 47-r+1\}$. If $k=r$, the position $j$ is selected uniformly at random from the set $\{1,\cdots,r-1\}$; otherwise the position $j$ is selected uniformly from the set $\{47-r+2,\cdots, \cdots, 47\}$. By this construction, $\calH(\vx)$ is a rank $r$ matrix. Moreover, it can be expressed as 
\begin{align*}
\calH(\vx) = \sqrt{w_j}\mG_j + \sqrt{w_k}\mG_k,
\end{align*}
where $w_j$ and $w_k$ are defined in \eqref{eq:wi}. Therefore, $\calH(\vx)$ is incoherent with the $k$-th Hankel basis, and only the average case incoherence property can be satisfied. The phase transition diagram for this case is presented in the right plot of Figure~\ref{fig:2D}, which shows that the phase transition curve  is still substantially high.

%%%%%%%%
\subsubsection{Phase transitions for 3D array completion}
\label{sec: 3dim}
The 3D array  that obeys the worst case incoherence condition is constructed as follows:
\begin{itemize}
	\item Generate a random 3D array $\tX^\natural\in\C^{9\times 9\times 16}$, where each frontal slice $\{ \tX_{\ell}^\natural \}_{\ell=1}^{16}$  of $\tX^\natural$ is a 2D random spectrally sparse signal. That is,
	\begin{align*}
	\tX_{\ell}^\natural(j,k) = \sum_{s=1}^{r}d_s w_s^j z_s^k,\quad (j,k)\in[9]\times [9],\numberthis\label{3D_test_signal}
	\end{align*} 
	where $w_s = e^{2\pi i f_{1 s}}$ and $z_s = e^{2\pi i f_{2 s}}$. Here, $f_{1s}, f_{2s}$ are uniformly sampled from $[0,1)$, and the complex coefficient is generated via $d_k=(1+10^{0.5c_k})e^{i\psi_k}$ with $\psi_k$ being uniformly sampled from $[0,2\pi)$  and $c_k$ being uniformly distributed on $[0,1]$.
	\item Apply the IDFT  to each tube of $\tX^\natural$ to get $\mX^\natural$, which gives the 3D test array.
	\end{itemize}  
To construct the 3D test array which only obeys the average case inherence condition,  two frontal slices of $\tX^\natural$ in \eqref{3D_test_signal} are replaced by the matrix $\calH(\vx)$, where $\vx\in\R^{49}$ is a vector with zero entries everywhere except for the two  positions $1\leq j,k\leq 49$. The positions of $j$ and $k$ are selected in the same way as that for the 2D case.  The phase transition diagrams corresponding to the two different types of 3D arrays are presented in Figure~\ref{fig:3D}, which exhibit a similar phenomenon to the 2D case.

\begin{figure}[ht!] \centering
	\subfigure{
		\label{fig c}
		\includegraphics[width=0.48\columnwidth]{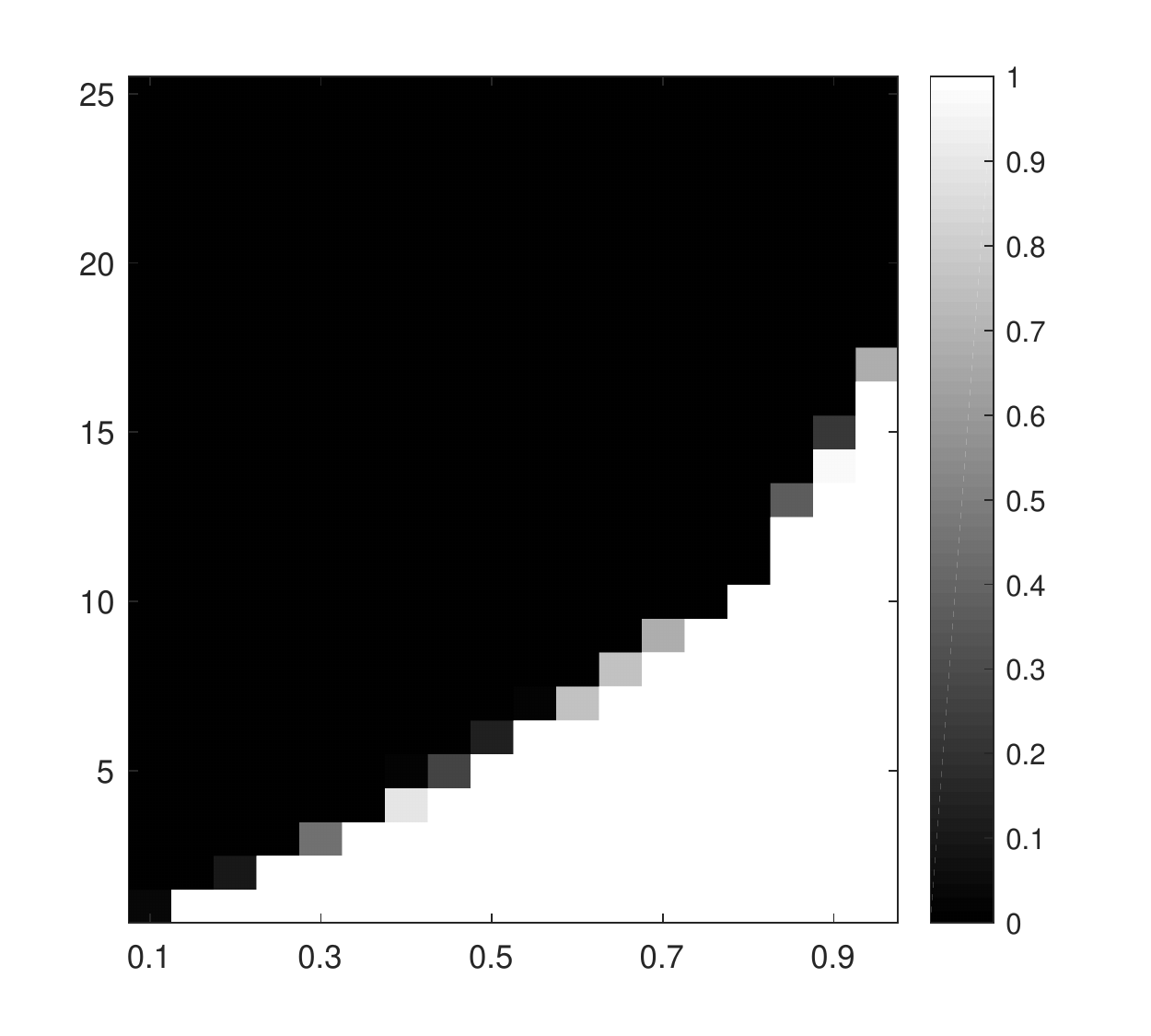}
	}
	\subfigure{
		\label{fig d}
		\includegraphics[width=0.48\columnwidth]{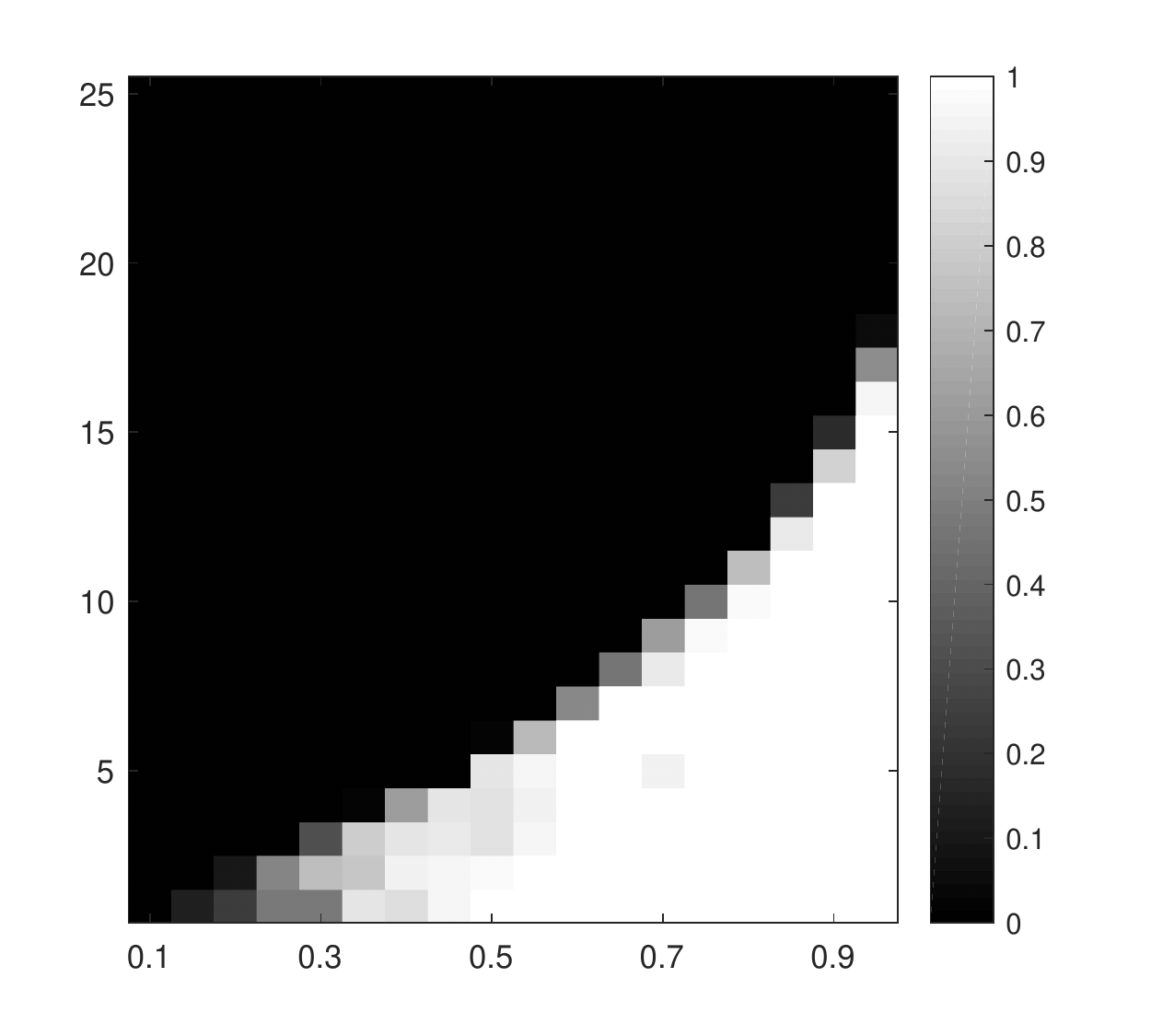}
	}
	
	\caption{Phase transition diagram for 3D array completion. The horizontal axis is  $p=m/dn$ and the vertical axis is $r$. Left: test array obeys the worst case incoherence condition; Right: test array obeys the average case incoherence condition but does not obey the worst case incoherence condition.}
	\label{fig:3D}
\end{figure}

\section{Proof outline of Theorem \ref{main result}}
\label{sec:proof1}
This section is devoted to the proof architecture  of Theorem \ref{main result}, while the proofs of the  intermediate results will be deferred to later sections. We follow a  route that has been well-established in the proofs of nuclear norm minimization  for various low rank matrix recovery problems \cite{candes2009exact,chen2015incoherence,chen2014robust}. To some extend,  the proof is an extension of the one for the $d=1$ case \cite{chen2014robust}. That being said, there are  two impediments of its own for the general $d$ case. Firstly, the properties of the DFT matrix play an important role in the analysis, which should be sufficiently utilized. Secondly, the analysis should be carried out very carefully to avoid the use of the worst case incoherence condition.%Compared with the proofs for the $d=1$ case in REF, a lot of care is needed when working out the analysis based on the average incoherence condition in REF rather than the worst incoherence condition in REF.
%%%%%
\subsection{Sufficient conditions for the dual certificate}
Similar to the  proof of Theorem \ref{main results: special}, in order to show that $\tZ^\natural$
is the optimal solution to \eqref{opti: nuclear norm}, we do not need to compare it with all the other feasible points, but only need to show the existence of  a dual certificate to certify the optimality of $\tZ^\natural$. The following lemma provides a set of  sufficient conditions for a valid dual certificate. 
\begin{lemma}
\label{lem:pf1_dual}
Suppose  $1\geq p\geq \frac{8}{n}$  and
	\begin{align}
	\label{ineq: RIP}
	\opnorm{\frac{1}{p}\calP_{\tT}\tcalG\calP_{\Omega}\tcalGT\calP_{\tT} -\calP_{\tT}\tcalG\tcalGT\calP_{\tT} } \leq\frac{1}{2}.
	\end{align}
If there exists a dual certificate $\bLambda=\diag(\bLambda_1,\cdots,\bLambda_d)\in\C^{dn_1\times dn_2}$ which obeys 
\begin{align}
	\label{eq: restricted on Omega}
	&\tcalG\calP_{\Omega}\tcalG^\ast(\bLambda)= \tcalG\tcalG^\ast(\bLambda),\\
	\label{ineq: F norm}
	&\fronorm{\calP_{\tT}(\bLambda) - \tU\tV^\tranH}\leq \frac{1}{n}, \\
	\label{ineq: perp operator norm}
	&\opnorm{\calP_{\tT^\perp}(\bLambda)} \leq \frac{1}{2},
	\end{align}
then the $d$-block diagonal matrix $\tZ^\natural$ is the unique optimal solution to \eqref{opti: nuclear norm}.
\end{lemma}
%%%%%
\begin{proof}
Consider any perturbation $\tZ^\natural + \tW$, where $\tW=\diag(\tW_1,\cdots,\tW_d)$ satisfies 
\begin{align*}
(\calI-\tcalG\tcalG^\ast)(\tW) = \bzero\quad\mbox{and}\quad\calP_{\Omega}\tcalG^\ast(\tW) = \bzero.
\end{align*}
Note that for any $\tW_i$, there exists  $\mS_i\in\tT_i^\perp$ such that
\begin{align*}
\la\mS_i, \calP_{\tT_i^\perp}{(\tW_i)} \ra = \nucnorm{\calP_{\tT_i^\perp}(\tW_i)} \quad\text{and}\quad\opnorm{\mS_i}\leq 1,
\end{align*}
which implies that $\tU_i\tV^\tranH_i + \mS_i\in\partial\nucnorm{\tZ^\natural_i}$. Thus we have
\begin{align*}
\sum_{i=1}^{d}\nucnorm{\tZ_i^\natural + \tW_i} &{\geq} \sum_{i=1}^{d}\left(\nucnorm{\tZ_i^\natural} + \left\la\tU_i\tV^\tranH_i + \mS_i, \tW_i \right\ra \right)\\
&=\sum_{i=1}^{d}\nucnorm{\tZ_i^\natural} + \sum_{i=1}^{d}\nucnorm{\calP_{\tT_i^\perp}(\tW_i)}  + \sum_{i=1}^{d} \left\la\mU_i\mV^\tranH_i, \tW_i \right\ra,
\end{align*}
where the last term can be split into two terms:
\begin{align*}
\sum_{i=1}^{d} \left\la\tU_i\tV^\tranH_i, \tW_i \right\ra &= \left\la \tU\tV^\tranH, \tW\right\ra =  \left\la \tU\tV^\tranH - \bLambda , \tW\right\ra + \left\la \bLambda , \tW\right\ra.\numberthis\label{eq:kweq01}
\end{align*}

The application of  \eqref{ineq: F norm} and \eqref{ineq: perp operator norm}
yields that
\begin{align*}
	\left\la \tU\tV^\tranH - \bLambda , \tW\right\ra &= \left\la \tU\tV^\tranH - \calP_{\tT}(\bLambda) , \tW\right\ra - \la \calP_{\tT^\perp}(\bLambda) ,\tW\ra  \\
	&\geq - \fronorm{\tU\tV^\tranH - \calP_{\tT}(\bLambda) }\cdot \fronorm{\calP_{\tT}(\tW)} - \opnorm{\calP_{\tT^\perp}(\bLambda)}\cdot \nucnorm{\calP_{\tT^\perp}(\tW)}\\
	&{\geq} - \frac{1}{n}\cdot \fronorm{\calP_{\tT}(\tW)} - \frac{1}{2}\cdot \nucnorm{\calP_{\tT^\perp}(\tW)}.\numberthis\label{eq:kweq02}
	\end{align*}
Additionally, the properties of $\tW$ imply that 
	\begin{align*}
	\left\la \bLambda , \tW\right\ra &=  \left\la (\calI-\tcalG\tcalG^\ast)\bLambda + \tcalG\tcalG^\ast(\bLambda) , \tW\right\ra \\
	&=\left\la (\calI-\tcalG\tcalG^\ast)\bLambda + \tcalG\calP_{\Omega}\tcalG^\ast(\bLambda) , \tW\right\ra\\
	&= 0.\numberthis\label{eq:kweq03}
	\end{align*}
Thus, after substituting \eqref{eq:kweq02} and \eqref{eq:kweq03} into \eqref{eq:kweq01}, we obtain
\begin{align*}
\sum_{i=1}^{d}\nucnorm{\tZ_i^\natural + \tW_i} &\geq \sum_{i=1}^{d}\nucnorm{\tZ_i^\natural} + \sum_{i=1}^{d}\nucnorm{\calP_{\tT_i^\perp}(\tW_i)} - \frac{1}{n}\cdot \fronorm{\calP_{\tT}(\tW)} - \frac{1}{2}\cdot \nucnorm{\calP_{\tT^\perp}(\tW)}\\
&=\sum_{i=1}^{d}\nucnorm{\tZ_i^\natural} + \frac{1}{2}\sum_{i=1}^{d}\nucnorm{\calP_{\tT_i^\perp}(\tW_i)} -\frac{1}{n}\cdot \fronorm{\calP_{\tT}(\tW)} \\
&\geq \sum_{i=1}^{d}\nucnorm{\tZ_i^\natural} + \frac{1}{2}\fronorm{\calP_{\tT^\perp}(\tW)} -\frac{1}{n}\cdot \fronorm{\calP_{\tT}(\tW)}.\end{align*}
Moreover, Lemma \ref{lem:app01} in Appendix \ref{sec Auxiliary} shows that under the condition \eqref{ineq: RIP} there holds 
	\begin{align*}
	\fronorm{\calP_{\tT}(\tW)} \leq \frac{2\sqrt{2}}{p}\fronorm{\calP_{\tT^\perp}(\tW)}.
	\end{align*}
Therefore we have 
\begin{align*}
\sum_{i=1}^{d}\nucnorm{\tZ_i^\natural + \tW_i} &\geq \sum_{i=1}^{d}\nucnorm{\tZ_i^\natural} + \frac{1}{2}\fronorm{\calP_{\tT^\perp}(\tW)} -\frac{1}{n}\cdot \frac{2\sqrt{2}}{p}\fronorm{\calP_{\tT^\perp}(\tW)}\\
&{\geq} \sum_{i=1}^{d}\nucnorm{\tZ_i^\natural},
\end{align*}
which certifies the optimality of \kw{$\tZ^\natural$}.

To show the uniqueness of \kw{$\tZ^\natural$}, note that the equality holds only when $\fronorm{\calP_{\tT^\perp}(\tW)}=0$, which implies $\tW = \calP_{\tT}(\tW)$. It follows that 
\begin{align*}
\fronorm{\calP_{\tT}(\tW)}^2 
&=  \la \calP_{\tT}(\tW), \tcalG\tcalG^\ast(\tW) \ra= \la \tW, \calP_{\tT}\tcalG\tcalG^\ast\calP_{\tT}(\tW) \ra\\
&= \la \tW, \left(\calP_{\tT}\tcalG\tcalG^\ast\calP_{\tT} - \frac{1}{p}\calP_{\tT}\tcalG\calP_{\Omega}\tcalGT\calP_{\tT} \right)(\tW) \ra + \la \tW, \left(\frac{1}{p}\calP_{\tT}\tcalG\calP_{\Omega}\tcalGT\calP_{\tT} \right)(\tW) \ra\\
&= \la \tW, \left(\calP_{\tT}\tcalG\tcalG^\ast\calP_{\tT} - \frac{1}{p}\calP_{\tT}\tcalG\calP_{\Omega}\tcalGT\calP_{\tT} \right)(\tW) \ra + \la \tW, \left(\frac{1}{p}\calP_{\tT}\tcalG\calP_{\Omega}\tcalGT\right)(\tW) \ra\\
&= \la \tW, \left(\calP_{\tT}\tcalG\tcalG^\ast\calP_{\tT} - \frac{1}{p}\calP_{\tT}\tcalG\calP_{\Omega}\tcalGT\calP_{\tT} \right)(\tW) \ra\\
&\leq \opnorm{\frac{1}{p}\calP_{\tT}\tcalG\calP_{\Omega}\tcalGT\calP_{\tT} -\calP_{\tT}\tcalG\tcalGT\calP_{\tT} }\cdot \fronorm{\tW}^2\\
&\leq \frac{1}{2}\fronorm{\calP_{\tT}(\tW)}^2,
\end{align*}
where the first line follows from $(\calI - \tcalG\tcalG^\ast)\tW = \bzero$, the forth line follows from $\calP_{\Omega}\tcalG^\ast(\tW) = \bzero$, and the last line follows from \eqref{ineq: RIP}. This implies that $\calP_{\tT}(\tW)=0$, so \kw{$\tZ^\natural$} is the unique minimizer.
\end{proof}
%%%%%
\subsection{Constructing  the dual certificate}
We will apply the golfing scheme to construct the dual certificate. The scheme was first proposed in \cite{candes2006robust} and then has become an indispensable tool in the analysis of convex relaxation methods for low complexity data recovery problems \cite{candes2009exact,recht2011simpler,gross2011recovering,chen2014robust,chen2015incoherence,vaiter2018model}. In a nutshell, the golfing scheme is a projected gradient iteration but using fresh measurements in each iteration.

To motivate the golfing scheme for constructing $\bLambda$, consider the following constrained least squares problem:
\begin{align*}
\minimize_{\tZ}~\fronorm{\calP_{\tT}(\tZ) - \tU\tV^\tranH}^2~\mbox{subject to}~\tZ= \tcalG\calP_{\Omega}\tcalG^\ast(\tZ)+(\calI-\tcalG\tcalG^\ast)(\tZ),
\end{align*}
which is formed from the conditions \eqref{ineq: F norm} and \eqref{eq: restricted on Omega}. Here we choose to neglect \eqref{ineq: perp operator norm} because if $\bLambda$ is sufficiently close to $\tU\tV^\tranH$, it is natural to expect \eqref{ineq: perp operator norm} holds simultaneously. Noting that
\begin{align*}
(\tcalG\calP_{\Omega}\tcalG^\ast)^2=\tcalG\calP_{\Omega}\tcalG^\ast, \quad
(\calI-\tcalG\tcalG^\ast)^2=\calI-\tcalG\tcalG^\ast,\quad\mbox{and }(\tcalG\calP_{\Omega}\tcalG^\ast)(\calI-\tcalG\tcalG^\ast)={0},
\end{align*}
a projected gradient method for the above optimization problem is given by
\kw{
\begin{align*}
\tZ^{k} &= \tZ^{k-1} + \left( \frac{1}{p}\tcalG\calP_{\Omega}\tcalG^\ast + (\calI- \tcalG\tcalG^\ast) \right)\calP_{\tT}\left(\tU\tV^\tranH - \calP_{\tT}(\tZ^{k-1})\right),
\end{align*}
}where $1/p$ is a rescaling constant for the first projection.  Due to the statistical dependence among the iterates, the convergence analysis of the above iteration is not easy. The golfing scheme proposes to break the statistical dependence by splitting $\Omega$ into a number of independent subsets and then using them sequentially. 

More precisely,  the golfing scheme for constructing a dual certificate satisfying \eqref{eq: restricted on Omega}--\eqref{ineq: perp operator norm} is given as follows:
\kw{
\begin{align*}
\tZ^0 &= \bzero\in\C^{dn_1\times dn_2},\\
\tZ^k &= \tZ^{k-1} + \left( \frac{1}{q}\tcalG\calP_{\Omega_k}\tcalG^\ast + (\calI- \tcalG\tcalG^\ast) \right)\calP_{\tT}(\tU\tV^\tranH - \calP_{\tT}(\tZ^{k-1})),\quad\text{ for }k=1,\cdots, k_0,\numberthis\label{eq:construct_dual}\\
\bLambda:&=\tZ^{k_0},
\end{align*}
}where $1/q$ is a rescaling constant whose value will become clear immediately. 
Inspired by the work in \cite{chen2015incoherence}, $\{\Omega_k\}_{k=1}^{k_0}$ can be constructed in the following way: Set $k_0 = \lceil2\log (dn)\rceil$, and then each $\Omega_k$ is sampled independently from each other according to
\begin{align*}
\Pr{(i,j)\in\Omega_k} = 1-(1-p)^{1/k_0}=:q.
\end{align*}
A simple calculation can show that $\Pr{(i,j)\in\bigcup_{k=1}^{k_0}\Omega_k}=p$, meaning that $\Omega$ and $\bigcup_{k=1}^{k_0}\Omega_k$ are identically distributed and we can instead consider the recovery problem with samples from $\bigcup_{k=1}^{k_0}\Omega_k$.
Therefore, when constructing the dual certificate via \eqref{eq:construct_dual}, different $\Omega_k$ can be used in each iteration. 
%%%%%
\subsection{Validating the dual certificate and completing the proof}
To complete the proof of Theorem \ref{main result}, we need to show that the dual certificate constructed from \eqref{eq:construct_dual} obeys the assumptions of Lemma \ref{lem:pf1_dual}. Towards this end, we first list several useful lemmas whose proofs will be presented in Sections \ref{proof lemma key01} and  \ref{proof lemma key02to05}.
%%%%
\begin{lemma}\label{lem:key01}
Suppose $\Omega$ is sampled according to the Bernoulli model.	If the sample complexity satisfies $p\gtrsim\frac{\mu_0r\log (dn)}{n}$, then
	\begin{align*}
	\opnorm{\frac{1}{p}\calP_{\tT}\tcalG\calP_{\Omega}\tcalGT\calP_{\tT} -\calP_{\tT}\tcalG\tcalGT\calP_{\tT}} \leq \frac{1}{2}
	\end{align*}
	holds with high probabilty.
	\end{lemma}
%%%%
\begin{lemma}\label{lem:key02}
Let $\tZ\in\C^{dn_1\times dn_2}$ be a fixed $d$-block diagonal matrix. Then 
\begin{align*}
	\opnorm{\left(\frac{1}{p}\tcalG\calP_{\Omega}\tcalG^\ast - \tcalG\tcalG^\ast\right)\tZ} \lesssim \sqrt{\frac{\log(dn)}{p}}\left\|\tZ\right\|_{\tcalG,\mathsf{F}} + \frac{\log(dn)}{p}\left\| \tZ \right\|_{\tcalG,\infty} 
	\end{align*}
holds with high probability. Here and throughout the paper  $\left\| \tZ \right\|_{\tcalG,\infty}$ and $\left\|\tZ\right\|_{\tcalG,\mathsf{F}}$  are defined as
\begin{align}
\label{def: G F norm}
&\left\| \tZ \right\|_{\tcalG,\mathsf{F}} := \sqrt{\sum_{(j,k)\in[d]\times [n]} \frac{1}{dw_k} \left\la\tZ, \tmG_{j,k} \right\ra^2 },\\
\label{def: G inf  norm}
&\left\| \tZ \right\|_{\tcalG,\infty} := \max_{(j,k)\in[d]\times [n]} \left| \frac{1}{\sqrt{d w_k}} \left\la\tZ, \tmG_{j,k} \right\ra  \right|.
\end{align}
\end{lemma}
%%%%
\begin{lemma}\label{lem:key03}
Let $\tZ\in\C^{dn_1\times dn_2}$ be a fixed $d$-block diagonal matrix.
Under the incoherence assumption \eqref{def: incoherent}, 
\begin{align*}
\left\| \left(\frac{1}{p}\calP_{\tT}\tcalG\calP_{\Omega}\tcalG^\ast - \calP_{\tT}\tcalG\tcalG^\ast \right)\tZ \right\|_{\tcalG,\mathsf{F}} \lesssim\sqrt{\frac{\mu_0\log(dn) r}{n}}\left( \sqrt{\frac{\log(dn)}{p} }\left\|\tZ\right\|_{\tcalG,\mathsf{F}}+\frac{\log(dn)}{p}  \left\| \tZ\right\|_{\tcalG, \infty} \right)
\end{align*}
holds with high probability.
\end{lemma}
%%%%
\begin{lemma}\label{lem:key04}
Let $\tZ\in\C^{dn_1\times dn_2}$ be a fixed $d$-block diagonal matrix.
Under the incoherence assumption \eqref{def: incoherent}, 
\begin{align*}
	\left\| \left(\frac{1}{p}\calP_{\tT}\tcalG\calP_{\Omega}\tcalG^\ast - \calP_{\tT}\tcalG\tcalG^\ast\right)\tZ\right\|_{\tcalG, \infty} \lesssim  \frac{\mu_0r}{n} \left( \sqrt{\frac{\log (dn)}{p}}\vecnorm{\tZ}{\tcalG,2} + \frac{\log (dn)}{p}\vecnorm{\tZ}{\tcalG,\infty} \right)
	\end{align*}
holds with high probability.
\end{lemma}
%%%%
\begin{lemma}\label{lem:key05}
Under the incoherence assumption \eqref{def: incoherent}, we have 
\begin{align*}
\left\|\tU\tV^\tranH\right\|_{\tcalG,\mathsf{F}}\lesssim \sqrt{\frac{\mu_0 r\log (dn)}{n}}
\quad\mbox{and}\quad \left\| \tU\tV^\tranH\right\|_{\tcalG, \infty}\lesssim \frac{\mu_0r}{n}.
\end{align*}
\end{lemma}
%%%%%%%
%%%%%%%
We are now in the position to prove Theorem~\ref{main result}.
\begin{proof}[Proof of Theorem \ref{main result}] %By Lemma \ref{lem:pf1_dual}, 
We only need to validate the assumptions in Lemma \ref{lem:pf1_dual}. Note that \eqref{ineq: RIP} follows from Lemma \ref{lem:key01}, and it is not hard to see that \eqref{eq: restricted on Omega} holds by the construction process \eqref{eq:construct_dual}. Thus it only remains to show \eqref{ineq: F norm} and \eqref{ineq: perp operator norm}. 

\paragraph{Validating \eqref{ineq: F norm}} Let \kw{$\tE_k = \tU\tV^\tranH - \calP_{\tT}(\tZ^{k})$}. Then a simple calculation yields that
\begin{align*}
\tE_k = \calP_{\tT}\left(\tcalG\tcalG^\ast  - \frac{1}{q}\tcalG\calP_{\Omega_k}\tcalG^\ast \right)\calP_{\tT}(\tE_{k-1}).
\end{align*}
If follows that 
\begin{align*}
\fronorm{\tE_k} 
&\leq \opnorm{\calP_{\tT}\left(\tcalG\tcalG^\ast  - \frac{1}{q}\tcalG\calP_{\Omega_k}\tcalG^\ast \right)\calP_{\tT}}\cdot \fronorm{\tE_{k-1}}\\
&\leq \frac{1}{2} \fronorm{\tE_{k-1}},
\end{align*}
where we have used Lemma \ref{lem:key01} in the last inequality by noting that $\Omega_k$ is independent of $\tE_{k-1}$ and $q=1-(1-p)^{1/k_0}\geq p/k_0\gtrsim \frac{\mu_0r\log (dn)}{n}$. Thus we have 
\begin{align*}
\fronorm{\calP_{\tT}(\bLambda) - \tU\tV^\tran} &= \fronorm{\tE_{k_0}}\leq \left(\frac{1}{2}\right)^{k_0}\fronorm{\tU\tV^\tran}\leq\frac{1}{(dn)^{2}}\fronorm{\tU\tV^\tran}= \frac{\sqrt{dr}}{(dn)^{2}}\leq \frac{1}{n}.
\end{align*}

\paragraph{Validating \eqref{ineq: perp operator norm}} Because 
\begin{align*}
\bLambda = \sum_{k=1}^{k_0}\left( \frac{1}{q}\tcalG\calP_{\Omega_k}\tcalG^\ast + (\calI - \tcalG\tcalG^\ast) \right)\tE_{k-1},
\end{align*}
we have 
\begin{align*}
\opnorm{\calP_{\tT^\perp}(\bLambda)} 
&= \opnorm{\sum_{k=1}^{k_0}\calP_{\tT^\perp}\left( \frac{1}{q}\tcalG\calP_{\Omega_k}\tcalG^\ast + (\calI - \tcalG\tcalG^\ast) \right)\tE_{k-1}}\\
&= \opnorm{\sum_{k=1}^{k_0}\calP_{\tT^\perp}\left( \frac{1}{q}\tcalG\calP_{\Omega_k}\tcalG^\ast  - \tcalG\tcalG^\ast \right)\tE_{k-1}}\\
&\leq \sum_{k=1}^{k_0}\opnorm{\calP_{\tT^\perp}\left( \frac{1}{q}\tcalG\calP_{\Omega_k}\tcalG^\ast  - \tcalG\tcalG^\ast \right)\tE_{k-1}},
\end{align*}
where the second line follows from the fact $\tE_{k-1}\in\tT$. 

Noticing $\tE_{k-1}$ is independent of $\Omega_k$, the application of Lemma \ref{lem:key02} gives that
\begin{align*}
\opnorm{\calP_{\tT^\perp}\left( \frac{1}{q}\tcalG\calP_{\Omega_k}\tcalG^\ast  - \tcalG\tcalG^\ast \right)\calP_{\tT}(\tE_{k-1})} &\leq \opnorm{\left( \frac{1}{q}\tcalG\calP_{\Omega_k}\tcalG^\ast  - \tcalG\tcalG^\ast \right)(\tE_{k-1})}\\
&\lesssim \sqrt{\frac{\log(dn)}{q}}\left\|\tE_{k-1}\right\|_{\tcalG,\mathsf{F}} + \frac{\log(dn)}{q}\left\| \tE_{k-1} \right\|_{\tcalG,\infty}.
\end{align*}
Moreover, by Lemmas \ref{lem:key03} and \ref{lem:key04} we have 
\begin{align*}
&\sqrt{\frac{\log(dn)}{q}}  \left\|\tE_{k-1}\right\|_{\tcalG,\mathsf{F}} + \frac{\log(dn)}{q} \left\| \tE_{k-1} \right\|_{\tcalG,\infty} \\
=& \sqrt{\frac{\log(dn)}{q}}  \left\| \calP_{\tT}\left(\tcalG\tcalG^\ast  - \frac{1}{q}\tcalG\calP_{\Omega_{k-1}}\tcalG^\ast \right)(\tE_{k-2}) \right\|_{\tcalG,\mathsf{F}} + \frac{\log(dn)}{q} \left\| \calP_{\tT}\left(\tcalG\tcalG^\ast  - \frac{1}{q}\tcalG\calP_{\Omega_{k-1}}\tcalG^\ast \right)(\tE_{k-2}) \right\|_{\tcalG,\infty}\\
\lesssim & \left(\sqrt{\frac{\mu_0 r\log^2(dn)}{qn}}+\frac{\mu_0r\log(dn)}{qn}\right) \left( \sqrt{\frac{\log(dn)}{q} }\left\|\tE_{k-2}\right\|_{\tcalG,\mathsf{F}}+\frac{\log(dn)}{q}  \left\| \tE_{k-2} \right\|_{\tcalG, \infty} \right)\\
\leq &\frac{1}{2}\left( \sqrt{\frac{\log(dn)}{q} }\left\|\tE_{k-2}\right\|_{\tcalG,\mathsf{F}}+\frac{\log(dn)}{q}  \left\| \tE_{k-2} \right\|_{\tcalG, \infty} \right),
%&\quad + \frac{\mu_0r\log(dn)}{qn}\left( \sqrt{\frac{\log (dn)}{q}}\vecnorm{\tE_{k-2}}{\tcalG,2} + \frac{\log (dn)}{q}\vecnorm{\tE_{k-2}}{\tcalG,\infty} \right)\\
%=& \left( c_3\sqrt{\frac{\mu_5 r\log(dn)}{qn}} +c_4\frac{\mu_6r\log(dn)}{qn}  \right)\left( \sqrt{\frac{\log(dn)}{q} }\left\|\tE_{k-2}\right\|_{\tcalG,\mathsf{F}}+\frac{\log(dn)}{q}  \left\| \tE_{k-2}\right\|_{\tcalG, \infty} \right),
\end{align*} 
where we have used the assumption $q\geq p/k_0\gtrsim \frac{\mu_0r\log^2 (dn)}{n}$ for sufficiently large constant in the last line. 
Applying this relation recursively yields that 
\begin{align*}
\opnorm{\calP_{\tT^\perp}\left( \frac{1}{q}\tcalG\calP_{\Omega_k}\tcalG^\ast  - \tcalG\tcalG^\ast \right)\tE_{k-1}}\leq \left(\frac{1}{2}\right)^{k-1}\left( \sqrt{\frac{\log(dn)}{q} }\left\|\tE_{0}\right\|_{\tcalG,\mathsf{F}}+\frac{\log(dn)}{q}  \left\| \tE_{0}\right\|_{\tcalG, \infty} \right).
\end{align*}
Finally we have 
\begin{align*}
\opnorm{\calP_{\tT^\perp}(\bLambda)} &\leq \sum_{k=1}^{k_0}\left(\frac{1}{2}\right)^{k-1}\left( \sqrt{\frac{\log(dn)}{q} }\left\|\tE_{0}\right\|_{\tcalG,\mathsf{F}}+\frac{\log(dn)}{q}  \left\| \tE_{0}\right\|_{\tcalG, \infty} \right)\\
&\lesssim \frac{1}{2} \left( \sqrt{\frac{\mu_0r\log(dn)}{qn} } +\frac{\mu_0r\log(dn)}{qn}  \right)\\
&\leq \frac{1}{2},
\end{align*}
where we have used Lemma \ref{lem:key05} in the second inequality.
\end{proof}

\section{Proof of Lemma~\ref{lem:key01}}
\label{proof lemma key01}
Lemma~\ref{lem:key01} not only appears in the assumptions of Lemma \ref{lem:pf1_dual}, but will also be used in the proof of \eqref{ineq: F norm}. This section presents a proof of this lemma.
\begin{proof}[Proof of Lemma~\ref{lem:key01}]
For any $d$-block diagonal matrix $\tW\in\C^{dn_1\times dn_2}$, we have 
\begin{align*}
\left(\frac{1}{p}\calP_{\tT}\tcalG\calP_{\Omega}\tcalGT\calP_{\tT} \right)\tW&=\sum_{j,k}\delta_{j,k}\left\la \tcalG^\ast\calP_{\tT}(\tW), \ve_j\ve_k^\tran \right\ra\left(\frac{1}{p}\calP_{\tT}\tcalG\right)(\ve_j\ve_k^\tran).
%&= :\sum_{j,k}\delta_{j,k} \mW_{j,k},
\end{align*}
Similarly, there holds 
\begin{align*}
\left(\calP_{\tT}\tcalG\tcalGT\calP_{\tT} \right)\tW = 
\sum_{j,k}\left\la \tcalG^\ast\calP_{\tT}(\tW), \ve_j\ve_k^\tran \right\ra\left(\calP_{\tT}\tcalG\right)(\ve_j\ve_k^\tran).
\end{align*}
Therefore we can rewrite $\left(\frac{1}{p}\calP_{\tT}\tcalG\calP_{\Omega}\tcalGT\calP_{\tT} -\calP_{\tT}\tcalG\tcalGT\calP_{\tT} \right)(\tW)$ as 
\begin{align*}
\left(\frac{1}{p}\calP_{\tT}\tcalG\calP_{\Omega}\tcalGT\calP_{\tT} -\calP_{\tT}\tcalG\tcalGT\calP_{\tT} \right)(\tW) = \sum_{j,k} \calZ_{j,k}(\tW),
\end{align*}
where 
\begin{align*}
\calZ_{j,k}:~\tW\rightarrow \left(\frac{1}{p}\delta_{j,k}-1\right)\left\la \tcalG^\ast\calP_{\tT}(\tW), \ve_j\ve_k^\tran \right\ra\left(\calP_{\tT}\tcalG\right)(\ve_j\ve_k^\tran) 
\end{align*}
is a self-adjoint operator since 
\begin{align*}
\left\la \calZ_{j,k}(\tW), \tZ \right\ra &= \left(\frac{1}{p}\delta_{j,k}-1\right)\left\la \tcalG^\ast\calP_{\tT}(\tW), \ve_j\ve_k^\tran \right\ra\cdot \left\la \left(\calP_{\tT}\tcalG\right)(\ve_j\ve_k^\tran)  , \tZ \right\ra\\
&=\left(\frac{1}{p}\delta_{j,k}-1\right)\left\la \tW , \calP_{\tT}\tcalG(\ve_j\ve_k^\tran) \right\ra\cdot \left\la \ve_j\ve_k^\tran , \tcalG^\ast\calP_{\tT}(\tZ) \right\ra\\
&= \left\la\tW, \calZ_{j,k}(\tZ) \right\ra.
\end{align*}
Therefore, we have
\begin{align*}
\opnorm{\frac{1}{p}\calP_{\tT}\tcalG\calP_{\Omega}\tcalGT\calP_{\tT} -\calP_{\tT}\tcalG\tcalGT\calP_{\tT} } = \opnorm{\sum_{j,k}\calZ_{j,k}}
\end{align*}
In order to apply the Bernstein inequality \eqref{bernstein} to bound the spectral norm, we need to bound $\opnorm{\calZ_{j,k}}$ and $\opnorm{\E{\sum_{j,k}\calZ_{j,k}^2}}$.

For the upper bound of $\opnorm{\calZ_{j,k}}$, a direct calculation yields that
\begin{align*}
	\opnorm{\calZ_{j,k}}  &= \sup_{ \fronorm{\tW}=1}\fronorm{ \calZ_{j,k}(\tW) }\\
	&\leq \frac{1}{p}\sup_{\fronorm{\tW}=1}\left|\left\la \tcalG^\ast\calP_{\tT}(\tW), \ve_j\ve_k^\tran \right\ra \right|\cdot \fronorm{\left(\calP_{\tT}\tcalG\right)(\ve_j\ve_k^\tran) }\\
	&= \frac{1}{p}\sup_{\fronorm{\tW}=1}\left|\left\la \tW, \calP_{\tT}\tcalG(\ve_j\ve_k^\tran) \right\ra \right|\cdot \fronorm{\left(\calP_{\tT}\tcalG\right)(\ve_j\ve_k^\tran) }\\
	&\leq \frac{1}{p} \fronorm{\left(\calP_{\tT}\tcalG\right)(\ve_j\ve_k^\tran) }^2\\
	&\leq \frac{1}{p}\frac{2\mu_0c_sr}{n},
	\end{align*}
	where the last inequality is due to \eqref{ineq 2}.

In order to bound $\opnorm{\E{\sum_{j,k}\calZ_{j,k}^2}}$, first note that
\begin{align*}
	\calZ_{j,k}^2(\tW) &= \calZ_{j,k}\left( \left(\frac{1}{p}\delta_{j,k}-1\right)\left\la \tcalG^\ast\calP_{\tT}(\tW), \ve_j\ve_k^\tran \right\ra \calP_{\tT}\tcalG (\ve_j\ve_k^\tran) \right)\\
	&=   \left(\frac{1}{p}\delta_{j,k}-1\right)\left\la \tcalG^\ast\calP_{\tT}(\tW), \ve_j\ve_k^\tran \right\ra \cdot \calZ_{j,k}\left( \calP_{\tT}\tcalG (\ve_j\ve_k^\tran) \right)\\
	&= \left(\frac{1}{p}\delta_{j,k}-1\right)^2\left\la \tcalG^\ast\calP_{\tT}(\tW), \ve_j\ve_k^\tran \right\ra \cdot \left\la\tcalG^\ast\calP_{\tT}\tcalG(\ve_j\ve_k^\tran), \ve_j\ve_k^\tran \right\ra \calP_{\tT}\tcalG (\ve_j\ve_k^\tran)\\
	&= \left(\frac{1}{p}\delta_{j,k}-1\right)^2\left\la \tcalG^\ast\calP_{\tT}(\tW), \ve_j\ve_k^\tran \right\ra \cdot \left\la\calP_{\tT}\tcalG(\ve_j\ve_k^\tran), \calP_{\tT}\tcalG(\ve_j\ve_k^\tran) \right\ra \calP_{\tT}\tcalG (\ve_j\ve_k^\tran)\\
	&= \left(\frac{1}{p}\delta_{j,k}-1\right)^2\left\la \tcalG^\ast\calP_{\tT}(\tW), \ve_j\ve_k^\tran \right\ra \cdot \fronorm{\calP_{\tT}\tcalG(\ve_j\ve_k^\tran)}^2\cdot \calP_{\tT}\tcalG (\ve_j\ve_k^\tran).
	\end{align*}
Hence,
	\begin{align*}
	\opnorm{\E{\sum_{j,k}\calZ_{j,k}^2}} &= \sup_{\fronorm{\tW}=1} \fronorm{\E{\sum_{j,k}\calZ_{j,k}^2(\tW)}}\\
	&\leq \frac{1}{p}\sup_{\fronorm{\tW}=1} \fronorm{\sum_{j,k}\left\la \tcalG^\ast(\tW_{\tT}), \ve_j\ve_k^\tran \right\ra \cdot \fronorm{\calP_{\tT}\tcalG(\ve_j\ve_k^\tran)}^2\cdot \calP_{\tT}\tcalG (\ve_j\ve_k^\tran) }\\
	&\leq \frac{1}{p}\max_{j,k}\fronorm{\calP_{\tT}\tcalG(\ve_j\ve_k^\tran)}^2 \cdot \sup_{\fronorm{\tW}=1} \fronorm{\sum_{j,k}\left\la \tcalG^\ast(\tW_{\tT}), \ve_j\ve_k^\tran \right\ra \cdot \calP_{\tT}\tcalG (\ve_j\ve_k^\tran)}\\
	&= \frac{1}{p}\max_{j,k}\fronorm{\calP_{\tT}\tcalG(\ve_j\ve_k^\tran)}^2 \cdot \sup_{ \fronorm{\tW}=1} \fronorm{ \calP_{\tT}\tcalG \left(\sum_{j,k}\left\la \tcalG^\ast(\tW_{\tT}), \ve_j\ve_k^\tran \right\ra  \ve_j\ve_k^\tran \right)}\\
	&\leq \frac{1}{p}\max_{j,k}\fronorm{\calP_{\tT}\tcalG(\ve_j\ve_k^\tran)}^2 \cdot \sup_{\fronorm{\tW}=1} \fronorm{  \sum_{j,k}\left\la \tcalG^\ast(\tW_{\tT}), \ve_j\ve_k^\tran \right\ra  \ve_j\ve_k^\tran }\\
	&= \frac{1}{p}\max_{j,k}\fronorm{\calP_{\tT}\tcalG(\ve_j\ve_k^\tran)}^2 \cdot \sup_{ \fronorm{\tW}=1} \fronorm{ \tcalG^\ast\calP_{\tT}(\tW) }\\
	&\leq \frac{1}{p}\max_{j,k}\fronorm{\calP_{\tT}\tcalG(\ve_j\ve_k^\tran)}^2 \\
	&\leq \frac{1}{p}\frac{2\mu_0c_sr}{n}.
	\end{align*}
	
Based on the above two bounds,  applying  the Bernstein inequality to  
$\opnorm{\sum_{j,k}\calZ_{j,k}}$ completes the proof of Lemma \ref{proof lemma key01}.
%%%%%%
\end{proof}
\section{Proofs of Lemmas~\ref{lem:key02} to \ref{lem:key05}}
\label{proof lemma key02to05}
In this section we present the proofs for Lemmas~\ref{lem:key02} to \ref{lem:key05}. These lemmas have been used when establishing the inequality \eqref{ineq: perp operator norm}.
%%%%
\subsection{Proof of Lemma~\ref{lem:key02}}
\begin{proof}[Proof of Lemma~\ref{lem:key02}]
Notice that
\begin{align*}
	\left(\frac{1}{p}\tcalG\calP_{\Omega}\tcalG^\ast - \tcalG\tcalG^\ast\right)\tZ &= \sum_{j,k}\left( \frac{1}{p}\delta_{j,k}-1\right)\la\tcalG^\ast\tZ, \ve_j\ve_k^\tran\ra\tcalG(\ve_j\ve_k^\tran)\\
	&=:\sum_{j,k}\tZ_{j,k},
	\end{align*}
	where $\tZ_{j,k}\in\C^{dn_1\times dn_2}$ are independent $d$-block diagonal matrices with zero mean. In order to prove Lemma~\ref{lem:key02} we only need to show that
\begin{align*}
\opnorm{\tZ_{j,k}}\leq \frac{1}{p}\left\| \tZ \right\|_{\tcalG,\infty}\quad\mbox{and}\quad
\max\left\{\opnorm{\E{\sum_{j,k}\tZ_{i,j}\tZ_{j,k}^\tranH}},\opnorm{\E{\sum_{j,k}\tZ_{i,j}^\tranH\tZ_{j,k}}}\right\}\leq \frac{1}{p}\left\|\tX\right\|_{\tcalG,\mathsf{F}}^2
\end{align*}
since the lemma then follows immediately from the Bernstein inequality \eqref{bernstein}. 

The operator norm of $\tZ_{j,k}$ can be bounded as follows
	\begin{align*}
	\opnorm{\tZ_{j,k}} &\leq \frac{1}{p}\left| \la\tcalG^\ast(\tZ), \ve_j\ve_k^\tran\ra \right|\cdot \opnorm{\tcalG(\ve_j\ve_k^\tran)}\\
	&\leq \frac{1}{p}\frac{1}{\sqrt{dw_k}} \left| \la\tcalG^\ast(\tZ), \ve_j\ve_k^\tran\ra \right|\\
	&\leq \frac{1}{p}\max_{(j,k)\in [d]\times[n]}\frac{1}{\sqrt{dw_k}} \left| \la\tcalG^\ast(\tZ), \ve_j\ve_k^\tran\ra \right|\\
	&= \frac{1}{p}\left\| \tZ \right\|_{\tcalG,\infty},
	\end{align*}
	where the second line follows from Lemma~\ref{lem:property 2}.

On the other hand, 
\begin{align*}
	\opnorm{\E{\sum_{j,k}\tZ_{i,j}\tZ_{j,k}^\tranH}} & \leq \frac{1}{p}\opnorm{\sum_{j,k} \left|\la\tcalG^\ast(\tZ), \ve_j\ve_k^\tran\ra \right|^2 \cdot \tcalG(\ve_j\ve_k^\tran)\left(\tcalG(\ve_j\ve_k^\tran) \right)^\tranH}\\
	&\leq \frac{1}{p}\sum_{j,k}  \left|\la\tcalG^\ast(\tZ), \ve_j\ve_k^\tran\ra \right|^2\cdot \opnorm{ \tcalG(\ve_j\ve_k^\tran)\left(\tcalG(\ve_j\ve_k^\tran) \right)^\tranH}\\
	&\leq \frac{1}{p}\sum_{j,k} \frac{1}{w_kd} \left|\la\tcalG^\ast(\tZ), \ve_j\ve_k^\tran\ra \right|^2\\
	&= \frac{1}{p}\left\|\tZ\right\|_{\tcalG,\mathsf{F}}^2,
	\end{align*}
	where the third line follows from Lemma~\ref{lem:property 2}.
	Similarly, we can obtain that $\opnorm{\E{\sum_{j,k}\tZ_{j,k}^\tranH\tZ_{i,j}}} \leq \frac{1}{p}\left\|\tZ\right\|_{\tcalG,\mathsf{F}} $, which completes the proof of Lemma~\ref{lem:key02}.
\end{proof}
%%%%
\subsection{Proof of Lemma~\ref{lem:key03}}
The following observation plays a vital role in the proof of Lemma~\ref{lem:key03} as well as in the proof of Lemma~\ref{lem:key05}. 
\begin{lemma}
\label{lem:key_obs01} 
For any $d$-block diagonal matrix $\tZ$ we have 
$$\left\|\tZ\right\|_{\tcalG,\mathsf{F}}^2=\frac{1}{d}\sum_{i=1}^{d}\sum_{k=1}^{n}\frac{1}{w_k}\left|\la \tZ_i, \mG_{k} \ra \right|^2.$$
\end{lemma}
\begin{proof}
This lemma follows from a direct calculation:
\begin{align*}
\left\| \tZ \right\|_{\tcalG,\mathsf{F}}&=\sum_{j=1}^{d}\sum_{k=1}^{n}\frac{1}{dw_k}\left| \la \tZ, \tmG_{j,k} \ra \right|^2\\
&=\sum_{k=1}^{n}\frac{1}{dw_k}\sum_{j=1}^{d}\left| \sum_{i=1}^{d}\la \tZ_i, (\ve_i^\tran\mF\ve_j)\mG_{k} \ra \right|^2\\
&=\sum_{k=1}^{n}\frac{1}{dw_k}\sum_{j=1}^{d}\sum_{i=1}^d\left|\la \tZ_i, (\ve_i^\tran\mF\ve_j)\mG_{k} \ra \right|^2+\sum_{k=1}^{n}\frac{1}{dw_k}\sum_{j=1}^{d}\sum_{i\neq \ell}\la \tZ_i, (\ve_i^\tran\mF\ve_j)\mG_{k} \ra\overline{\la \tZ_\ell, (\ve_\ell^\tran\mF\ve_j)\mG_{k} \ra}\\
&=\sum_{k=1}^{n}\frac{1}{dw_k}\sum_{j=1}^{d}\sum_{i=1}^d\left|\la \tZ_i, (\ve_i^\tran\mF\ve_j)\mG_{k} \ra \right|^2+\sum_{k=1}^{n}\frac{1}{dw_k}\sum_{i\neq \ell}\left(\sum_{j=1}^{d}\overline{(\ve_i^\tran\mF\ve_j)}(\ve_\ell^\tran\mF\ve_j)\right)\la \tZ_i, \mG_{k} \ra\overline{\la \tZ_\ell, \mG_{k} \ra}\\
&=\sum_{k=1}^{n}\frac{1}{dw_k}\sum_{j=1}^{d}\sum_{i=1}^d\left|\la \tZ_i, (\ve_i^\tran\mF\ve_j)\mG_{k} \ra \right|^2\\
&=\sum_{k=1}^{n}\frac{1}{dw_k}\sum_{i=1}^d\left|\la \tZ_i, \mG_{k} \ra \right|^2,\end{align*}
where in the second to  last line we have used the fact that $\sum_{j=1}^{d}\overline{(\ve_i^\tran\mF\ve_j)}(\ve_\ell^\tran\mF\ve_j)=0$ since $\mF$ is a unitary matrix, and the last line follows from the fact $|\ve_i^\tran\mF\ve_j|^2=1/d$.
\end{proof}
%%%%%
\begin{lemma}\label{lem:key_obs01a}
For any pair of $(\alpha,\beta)\in[d]\times[n]$ we have
\begin{align*}
\vecnorm{ \calP_{\tT}\left( \sqrt{dw_\beta}\tmG_{\alpha,\beta} \right) }{\tcalG,\mathsf{F}}^2\lesssim
 \sqrt{\frac{\mu_0 r\log(dn)}{n}}
\end{align*}
\begin{proof}
Noting that the $i$th block of $\calP_{\tT}\left( \sqrt{dw_\beta}\tmG_{\alpha,\beta} \right)$ is $\calP_{T_i}\left( \sqrt{dw_\beta}\ve_i^\tran\mF\ve_\alpha\mG_{\beta} \right)$,
it follows from Lemma~\ref{lem:key_obs01a} that
\begin{align*}
\vecnorm{ \calP_{\tT}\left( \sqrt{dw_\beta}\tmG_{\alpha,\beta} \right) }{\tcalG,\mathsf{F}}^2=\frac{1}{d}\sum_{i=1}^{d}\sum_{k=1}^{n}\frac{1}{w_k}\left|\la \calP_{\tT_i}\left( \sqrt{dw_\beta}\ve_i^\tran\mF\ve_\alpha\mG_{\beta} \right), \mG_{k} \ra \right|^2.
\end{align*}
For any $k$, a direct calculation yields that 
\begin{align*}
&\quad \frac{1}{d}\sum_{i=1}^{d}\vecnorm{\ve_k^\tran\calP_{\tT_i}\left(\sqrt{dw_\beta}\ve_i^\tran\mF\ve_\alpha\mG_{\beta} \right) }{2}^2 \\
&= \frac{w_\beta}{d}\sum_{i=1}^{d}\vecnorm{\ve_k^\tran\left( \tU_i\tU_i^\tranH\mG_{\beta} + \mG_{\beta}\tV_i\tV_i^\tranH - \tU_i\tU_i^\tranH\mG_{\beta}\tV_i\tV_i^\tranH \right) }{2}^2 \\
&\leq 3\frac{w_\beta}{d}\sum_{i=1}^{d}\vecnorm{\ve_k^\tran\tU_i\tU_i^\tranH\mG_{\beta}}{2}^2
	 + 3\frac{w_\beta}{d}\sum_{i=1}^{d}\vecnorm{\ve_k^\tran\mG_{\beta}\tV_i\tV_i^\tranH }{2}^2
	 + 3\frac{w_\beta}{d}\sum_{i=1}^{d}\vecnorm{\ve_k^\tran\tU_i\tU_i^\tranH\mG_{\beta}\tV_i\tV_i^\tranH }{2}^2\\
&\leq \frac{9\mu_0r}{n},
\end{align*}
where the last line follows from the fact $\|\mG_\beta\|\leq 1/\sqrt{w_\beta}$, the inequality 
\begin{align*}
	\frac{w_\beta}{d}\sum_{i=1}^{d}\vecnorm{\ve_k^\tran\mG_{\beta}\tV_i\tV_i^\tranH }{2}^2 &= \frac{w_\beta}{d}\sum_{i=1}^{d}\vecnorm{\frac{1}{\sqrt{w_\beta}}\sum_{p+q=\beta}\ve_k^\tran\ve_p\ve_q^\tran\tV_i\tV_i^\tranH }{2}^2\\
	& = \frac{1}{d}\sum_{i=1}^{d}\vecnorm{\ve_{\beta-k}^\tran\tV_i\tV_i^\tranH }{2}^2  \quad\mbox{or}\quad 0\\
	&\leq \frac{\mu_0r}{n},
	\end{align*}
and the incoherence condition.  Then the application of Lemma \ref{lemma : general} concludes the proof.
\end{proof}
\end{lemma}
%%%%%
\begin{proof}[Proof of Lemma~\ref{lem:key03}] 
Define 
\begin{align*}
	z^{j,k} &= \frac{1}{\sqrt{dw_k}} \left\la\left(\frac{1}{p}\calP_{\tT}\tcalG\calP_{\Omega}\tcalG^\ast - \calP_{\tT}\tcalG\tcalG^\ast \right)(\tZ), \tcalG(\ve_j\ve_k^\tran) \right\ra\\
	&= \frac{1}{\sqrt{dw_k}}\sum_{\alpha=1}^{d}\sum_{\beta=1}^{n}\left(\frac{1}{p}\delta_{\alpha,\beta} - 1\right) \left\la \tcalG^\ast(\tZ), \ve_\alpha\ve_{\beta}^\tran \right\ra\cdot \left\la \calP_{\tT}\tcalG(\ve_\alpha\ve_\beta^\tran), \tcalG(\ve_j\ve_k^\tran) \right\ra, 
	\end{align*}
and let $\vz=[z^{1,1},\cdots,z^{d,n}]^\tran\in\C^{dn}$. Then one can easily see that
\begin{align*}
	\left\| \left(\frac{1}{p}\calP_{\tT}\tcalG\calP_{\Omega}\tcalG^\ast - \calP_{\tT}\tcalG\tcalG^\ast \right)(\tZ) \right\|_{\tcalG,\mathsf{F}} &=\sqrt{\sum_{j,k} \frac{1}{(dw_k)}\left|\left\la \left(\frac{1}{p}\calP_{\tT}\tcalG\calP_{\Omega}\tcalG^\ast -  \calP_{\tT}\tcalG\tcalG^\ast\right)(\tZ), \tcalG(\ve_j\ve_k^\tran)\right\ra \right|^2 }=\vecnorm{\vz}{}.
	\end{align*}
Moreover, if we define $\vz_{\alpha,\beta}\in\C^{dn}$ as
\begin{align*}
\vz_{\alpha,\beta} &= \left(\frac{1}{p}\delta_{\alpha,\beta} - 1\right) \left\la \tcalG^\ast(\tZ), \ve_\alpha\ve_{\beta}^\tran \right\ra\cdot \begin{bmatrix}
	\frac{1}{\sqrt{dw_1}}\left\la \calP_{\tT}\tcalG(\ve_\alpha\ve_\beta^\tran), \tcalG(\ve_1\ve_1^\tran) \right\ra \\
	\vdots\\
	\frac{1}{\sqrt{dw_k}}\left\la \calP_{\tT}\tcalG(\ve_\alpha\ve_\beta^\tran), \tcalG(\ve_j\ve_k^\tran) \right\ra\\
	\vdots\\
	\frac{1}{\sqrt{dw_n}}\left\la \calP_{\tT}\tcalG(\ve_\alpha\ve_\beta^\tran), \tcalG(\ve_d\ve_n^\tran) \right\ra 
	\end{bmatrix}=:\left(\frac{1}{p}\delta_{\alpha,\beta} - 1\right)\vs_{\alpha,\beta},
\end{align*}
then it follows that
\begin{align*}
	\left\| \left(\frac{1}{p}\calP_{\tT}\tcalG\calP_{\Omega}\tcalG^\ast - \calP_{\tT}\tcalG\tcalG^\ast \right)(\tZ) \right\|_{\tcalG,\mathsf{F}} &= \vecnorm{\sum_{\alpha,\beta}\vz_{\alpha,\beta} }{}.
	\end{align*}

Firstly, $\opnorm{\vz_{\alpha,\beta}}$ can be bounded as follows
\begin{align*}
\opnorm{\vz_{\alpha,\beta}}& \leq \frac{1}{p}\left| \left\la \tcalG^\ast(\tZ), \ve_\alpha\ve_{\beta}^\tran \right\ra \right|  \cdot \sqrt{ \sum_{j,k} \left( \frac{1}{\sqrt{dw_k}} \left| \left\la \calP_{\tT}\tcalG(\ve_\alpha\ve_\beta^\tran), \tcalG(\ve_j\ve_k^\tran) \right\ra \right| \right)^2 }\\
	&= \frac{1}{p}{\left| \left\la \tcalG^\ast(\tZ), \ve_\alpha\ve_{\beta}^\tran \right\ra \right| }\left\|\calP_{\tT}\tcalG(\ve_\alpha\ve_\beta^\tran) \right\|_{\tcalG,\mathsf{F}}\\
	%&=\frac{1}{p}{\left| \left\la \tcalG^\ast(\tZ), \ve_\alpha\ve_{\beta}^\tran \right\ra \right| }\sqrt{\frac{1}{d}\sum_{i=1}^{d}\sum_{k=1}^{n}\frac{1}{w_k}\left|\la [\calP_{\tT}\tcalG(\ve_\alpha\ve_\beta^\tran)]_i, \mG_{k} \ra \right|^2}\\
	&\lesssim \frac{1}{p}\frac{\left| \left\la \tcalG^\ast(\tZ), \ve_\alpha\ve_{\beta}^\tran \right\ra \right| }{\sqrt{dw_{\beta}}} \cdot \sqrt{\frac{\mu_0 r\log(dn)}{n}}\\
	&\lesssim \frac{1}{p} \sqrt{\frac{\mu_0 r\log(dn)}{n}} \left\| \tZ\right\|_{\tcalG, \infty},
	\end{align*}
	where the third line follows from Lemma~\ref{lem:key_obs01a}.
	
Secondly, we have 
\begin{align*}
\opnorm{\E{\sum_{\alpha,\beta}\vz_{\alpha,\beta}\vz_{\alpha,\beta}^\tranH}} &\leq \frac{1}{p}\sum_{\alpha,\beta} \vecnorm{\vs_{\alpha,\beta}}{2}^2 = \frac{1}{p}\sum_{\alpha,\beta}\left| \left\la \tcalG^\ast(\tZ), \ve_\alpha\ve_{\beta}^\tran \right\ra \right|^2  \cdot \left\|\calP_{\tT}\tcalG(\ve_\alpha\ve_\beta^\tran) \right\|_{\tcalG,\mathsf{F}}^2\\
%&= \frac{1}{p}\sum_{\alpha,\beta}{\left| \left\la \tcalG^\ast(\tZ), \ve_\alpha\ve_{\beta}^\tran \right\ra \right|^2} \lb\frac{1}{d}\sum_{i=1}^{d}\sum_{k=1}^{n}\frac{1}{w_k}\left|\la [\calP_{\tT}\rb\tcalG(\ve_\alpha\ve_\beta^\tran)]_i, \mG_{k} \ra \right)\\
&= \frac{1}{p}\sum_{\alpha,\beta}{\left| \left\la \tcalG^\ast(\tZ), \ve_\alpha\ve_{\beta}^\tran \right\ra \right|^2} \left(\frac{1}{d}\sum_{i=1}^{d}\sum_{k=1}^{n}\frac{1}{w_k}\left|\la [\calP_{\tT}\tcalG(\ve_\alpha\ve_\beta^\tran)]_i, \mG_{k} \ra \right|^2\right)\\
	&\lesssim \frac{\mu_0 r\log(dn)}{pn}\sum_{\alpha,\beta}\frac{\left| \left\la \tcalG^\ast(\tZ), \ve_\alpha\ve_{\beta}^\tran \right\ra \right|^2}{dw_{\beta}} \\
	&= \frac{\mu_0 r\log(dn)}{pn}\cdot \left\|\tZ\right\|_{\tcalG,\mathsf{F}}^2.
	\end{align*}
	Similarly, there holds $\opnorm{\E{\sum_{\alpha,\beta}\vz_{\alpha,\beta}^\tranH\vz_{\alpha,\beta}}}\lesssim \frac{\mu_0 r\log(dn)}{pn}\cdot \left\|\tZ\right\|_{\tcalG,\mathsf{F}}^2$.
	
	Finally, applying the Bernstein inequality to $\vecnorm{\sum_{\alpha,\beta}\vz_{\alpha,\beta}}{}$ completes the proof.
\end{proof}
%%%%
\subsection{Proof of Lemma~\ref{lem:key04}}
The following lemma will be used in the proof of Lemma~\ref{lem:key04}.
\begin{lemma}\label{lem:key_obs02}
For any two pairs of $(j,k), (\alpha,\beta)\in[d]\times[n]$, there holds
	\begin{align*}
	\sqrt{\frac{w_k}{w_\beta}}\left| \la \calP_{\tT}(\tmG_{j,k}), \tmG_{\alpha,\beta} \ra \right| &\leq \frac{3\mu_0r}{n}.
	\end{align*}
\end{lemma}
\begin{proof}
By the definition of $\calP_{\tT}$ there holds
	\begin{align*}
	\sqrt{\frac{w_k}{w_\beta}}\left| \la \calP_{\tT}(\tmG_{j,k}), \tmG_{\alpha,\beta} \ra \right| & = \sqrt{\frac{w_k}{w_\beta}}\left| \la \tU\tU^\tranH\tmG_{j,k}+ \tmG_{j,k}\tV\tV^\tranH -\tU\tU^\tranH\tmG_{j,k}\tV\tV^\tranH , \tmG_{\alpha,\beta} \ra \right| \\
	&\leq \sqrt{\frac{w_k}{w_\beta}}\left| \la \tU\tU^\tranH\tmG_{j,k}, \tmG_{\alpha,\beta}  \ra\right| 
	+ \sqrt{\frac{w_k}{w_\beta}}\left| \la \tmG_{j,k}\tV\tV^\tranH ,\tmG_{\alpha,\beta} \ra \right|\\
	&
	+ \sqrt{\frac{w_k}{w_\beta}}\left| \la \tU\tU^\tranH\tmG_{j,k}\tV\tV^\tranH , \tmG_{\alpha,\beta} \ra \right|.
	\end{align*}
Then it suffices to bound each of the above three terms separately. 

	For the first term we have 
	\begin{align*}
	\sqrt{\frac{w_k}{w_\beta}}\left| \la \tU\tU^\tranH(\tmG_{j,k}), \tmG_{\alpha,\beta} \ra \right|
	&= \sqrt{\frac{w_k}{w_\beta}}\left| \sum_{i=1}^{d}(\ve_i^\tran\mF\ve_j)\cdot \overline{(\ve_i^\tran\mF\ve_{\alpha})}\cdot \la \tU_i{\tU_i}^\tranH\mG_k, \mG_{\beta} \ra \right| \\
	&\leq \frac{1}{d}\sum_{i=1}^{d} \sqrt{\frac{w_k}{w_\beta}}\left| \la \tU_i\tU_i^\tranH\mG_k, \mG_{\beta} \ra \right|\\
	&= \frac{1}{d}\sum_{i=1}^{d} \sqrt{\frac{w_k}{w_\beta}}\left| \frac{1}{\sqrt{w_k\cdot w_{\beta}}}\sum_{p+q=k+1}\sum_{a+b=\beta+1}\la \tU_i\tU_i^\tranH\ve_p\ve_q^\tran, \ve_a\ve_b^\tran \ra \right|\\
	&= \frac{1}{d}\sum_{i=1}^{d} \frac{1}{w_\beta}\left|\sum_{p+q=k+1}\sum_{a+b=\beta+1}\la \tU_i^\tranH\ve_p, \tU_i^\tranH\ve_a\ve_b^\tran\ve_q \ra \right|\\
	&= \frac{1}{d}\sum_{i=1}^{d} \frac{1}{w_\beta}\left| \sum_{a+b=\beta+1,b\leq k}\la \tU_i^\tranH\ve_{k-b}, \tU_i^\tranH\ve_a\ra \right|\\
	& \leq \frac{1}{dw_\beta} \sum_{i=1}^{d}\sum_{a+b=\beta,b\leq k} \vecnorm{\tU_i^\tranH\ve_{k-b} }{2} \cdot \vecnorm{ \tU_i^\tranH\ve_a}{2}\\
	&\leq \sqrt{\frac{1}{dw_\beta} \sum_{i=1}^{d}\sum_{a+b=\beta+1,b\leq k} \vecnorm{\tU_i^\tranH\ve_{a} }{2}^2} \cdot \sqrt{\frac{1}{dw_\beta} \sum_{i=1}^{d}\sum_{a+b=\beta+1,b\leq k} \vecnorm{\tU_i^\tranH\ve_{k-b} }{2}^2}\\
	&= \sqrt{\frac{1}{w_\beta}\sum_{a+b=\beta+1,b\leq k}  \frac{1}{d}\sum_{i=1}^{d}\vecnorm{\tU_i^\tranH\ve_{a} }{2}^2}\sqrt{\frac{1}{w_\beta}\sum_{a+b=\beta+1,b\leq k}  \frac{1}{d}\sum_{i=1}^{d}\vecnorm{\tU_i^\tranH\ve_{k-b} }{2}^2}\\
	&\leq \frac{\mu_0r}{n}.
	\end{align*}
	Additionally, the second term can be bounded in a similar way.
	
	For the third term we have 
	\begin{align*}
	\sqrt{\frac{w_k}{w_{\beta}}} \left| \la \tU\tU^\tranH\tmG_{j,k}, \tmG_{\alpha,\beta}\tV\tV^\tranH \ra \right|&= \sqrt{\frac{w_k}{w_{\beta}}} \left| \sum_{i=1}^{d}(\ve_i^\tran\mF\ve_j)\overline{(\ve_i^\tran\mF\ve_{\alpha})}\la \tU_i\tU_i^\tranH\mG_k, \mG_{\beta}\tV_i\tV_i^\tranH \ra \right|\\
	&\leq \frac{1}{d}\sqrt{\frac{w_k}{w_{\beta}}}\sum_{i=1}^{d} \left| \la \tU_i\tU_i^\tranH\mG_k, \mG_{\beta}\tV_i\tV_i^\tranH \ra \right|\\
	&= \frac{1}{d}\sqrt{\frac{w_k}{w_{\beta}}}\sum_{i=1}^{d} \left| \frac{1}{\sqrt{w_k\cdot w_\beta}}\sum_{p+q=k+1} \sum_{a+b=\beta+1}\la \tU_i\tU_i^\tranH\ve_p\ve_q^\tran, \ve_a\ve_b^\tran\tV_i\tV_i^\tranH \ra \right|\\
	&\leq \frac{1}{d}\sqrt{\frac{w_k}{w_{\beta}}}\sum_{i=1}^{d}\frac{1}{\sqrt{w_k\cdot w_\beta}}\sum_{p+q=k+1} \sum_{a+b=\beta+1} \left| \la \tU_i\tU_i^\tranH\ve_p\ve_q^\tran, \ve_a\ve_b^\tran\tV_i\tV_i^\tranH \ra \right|\\
	&\leq \frac{1}{d w_\beta} \sum_{i=1}^{d} \sum_{p+q=k+1} \sum_{a+b=\beta+1} \left|  \ve_a^\tran\tU_i\tU_i^\tranH\ve_p\right|\cdot \left| \ve_b^\tran\tV_i\tV_i^\tranH\ve_q  \right|\\
	&\leq \frac{1}{d w_\beta}\sqrt{\sum_{i=1}^{d} \sum_{p+q=k+1}\sum_{a+b=\beta+1}\left|  \ve_a^\tran\tU_i\tU_i^\tranH\ve_p\right|^2 }  \sqrt{\sum_{i=1}^{d} \sum_{p+q=k+1}\sum_{a+b=\beta+1}\left|  \ve_b^\tran\tV_i\tV_i^\tranH\ve_q\right|^2 }\\
	&\leq \frac{1}{d w_\beta} \sqrt{\sum_{i=1}^{d} \sum_{p=1}^{n_1}\sum_{a+b=\beta}\left|  \ve_a^\tran\tU_i\tU_i^\tranH\ve_p\right|^2 } \cdot \sqrt{\sum_{i=1}^{d} \sum_{q=1}^{n_2}\sum_{a+b=\beta}\left|  \ve_b^\tran\tV_i\tV_i^\tranH\ve_q\right|^2 }\\
	&= \sqrt{\frac{1}{d w_\beta} \sum_{i=1}^{d}  \sum_{a+b=\beta} \left\| \ve_a^\tran\tU_i\tU_i^\tranH \right\|_2^2}\cdot \sqrt{\frac{1}{d w_\beta} \sum_{i=1}^{d}  \sum_{a+b=\beta}\left\|\ve_b^\tran\tV_i\tV_i^\tranH\right\|_2^2 }\\
	&= \sqrt{\frac{1}{d } \sum_{i=1}^{d}  \left\| \mG_\beta^\tran\tU_i\tU_i^\tranH \right\|_F^2}\cdot \sqrt{\frac{1}{d} \sum_{i=1}^{d} \left\|\mG_{\beta}^\tran\tV_i\tV_i^\tranH\right\|_F^2 }\\
	&\leq \frac{\mu_0r}{n}
	\end{align*}
	
	Combining the above bounds together completes the proof the lemma.
\end{proof}

\begin{proof}[Proof of Lemma~\ref{lem:key04}]
Since 
\begin{align*}
	\left(\frac{1}{p}\calP_{\tT}\tcalG\calP_{\Omega}\tcalG^\ast - \calP_{\tT}\tcalG\tcalG^\ast\right)(\tZ) &= \sum_{j,k}\left(\frac{1}{p}\delta_{j,k} - 1\right) \left\la \tcalG^\ast(\tZ), \ve_j\ve_k^\tran \right\ra \calP_{\tT}\tcalG(\ve_j\ve_k^\tran),
	\end{align*}
	we have 
	\begin{align*}
	\left\| \left(\frac{1}{p}\calP_{\tT}\tcalG\calP_{\Omega}\tcalG^\ast - \calP_{\tT}\tcalG\tcalG^\ast\right)(\tZ)\right\|_{\tcalG, \infty} &= \max_{\alpha,\beta} \left| \frac{1}{\sqrt{dw_{\beta}}}\left\la  \sum_{j,k}\left(\frac{1}{p}\delta_{j,k} - 1\right) \left\la \tcalG^\ast(\tZ), \ve_j\ve_k^\tran \right\ra \calP_{\tT}\tcalG(\ve_j\ve_k^\tran), \tcalG(\ve_{\alpha}\ve_{\beta}^\tran )\right\ra \right|\\
	&=  \max_{\alpha,\beta} \left| \frac{1}{\sqrt{dw_{\beta}}} \sum_{j,k} \left(\frac{1}{p}\delta_{j,k} - 1\right) \left\la \tcalG^\ast(\tZ), \ve_j\ve_k^\tran \right\ra\cdot \left\la   \calP_{\tT}\tcalG(\ve_j\ve_k^\tran), \tcalG(\ve_{\alpha}\ve_{\beta}^\tran )\right\ra \right|\\
	&=: \max_{\alpha,\beta} \left| \sum_{j,k} Z_{j,k}^{\alpha,\beta}\right|.
	\end{align*}
	
	For any fixed pair of $(\alpha,\beta)$, $X_{j,k}^{\alpha,\beta}$ are mean-zero independent random variables with 
	\begin{align*}
	\left| Z_{j,k}^{\alpha,\beta} \right| &\leq \frac{1}{p}\cdot \left|\left\la \tcalG^\ast(\tZ), \ve_j\ve_k^\tran \right\ra \right|\cdot \frac{1}{\sqrt{dw_{\beta}}}  \left|\left\la   \calP_{\tT}\tcalG(\ve_j\ve_k^\tran), \tcalG(\ve_{\alpha}\ve_{\beta}^\tran )\right\ra \right|\\
	& = \frac{1}{p}\cdot \frac{\left|\left\la \tcalG^\ast(\tZ), \ve_j\ve_k^\tran \right\ra \right|}{\sqrt{dw_k}}\cdot \frac{\sqrt{dw_k}}{\sqrt{dw_{\beta}}}  \left|\left\la   \calP_{\tT}\tcalG(\ve_j\ve_k^\tran), \tcalG(\ve_{\alpha}\ve_{\beta}^\tran )\right\ra \right|\\
	&\lesssim \frac{1}{p}\cdot \frac{\left|\left\la \tcalG^\ast(\tZ), \ve_j\ve_k^\tran \right\ra \right|}{\sqrt{dw_k}}\cdot \frac{\mu_0r}{n} \\
	&\lesssim \frac{1}{p}\frac{\mu_0r}{n} \vecnorm{\tZ}{\tcalG,\infty},
	\end{align*}
	where the third line follows from Lemma~\ref{lem:key_obs02}. Moreover,	\begin{align*}
	\E {\sum_{j,k}\left| Z_{j,k}^{\alpha,\beta} \right|^2}&\lesssim \frac{1}{p}\sum_{j,k}\frac{\left|\left\la \tcalG^\ast(\tZ), \ve_j\ve_k^\tran \right\ra \right|^2}{dw_k}\cdot \left(\frac{\mu_0r}{n}\right)^2 \\
	&\lesssim \frac{1}{p}\left(\frac{\mu_0r}{n}\right)^2\cdot \vecnorm{\tZ}{\tcalG,\mathsf{F}}^2.
	\end{align*}
	Thus, by the Bernstein inequality \eqref{bernstein},
	\begin{align*}
	\left| \sum_{j,k} Z_{j,k}^{\alpha,\beta}\right| \lesssim\frac{\mu_0r}{n} \left( \sqrt{\frac{\log (dn)}{p}}\vecnorm{\tZ}{\tcalG,\mathsf{F}} + \frac{\log (dn)}{p}\vecnorm{\tZ}{\tcalG,\infty} \right) 
	\end{align*}
	holds with high probability. 
	
	Since there are only $dn$ pairs of $(\alpha,\beta)$, taking a union bound concludes the proof.
\end{proof}
%%%%
\subsection{Proof of Lemma~\ref{lem:key05}}
\begin{proof}[Proof of Lemma~\ref{lem:key05}]
By Lemma~\ref{lem:key_obs01} we have 
\begin{align*}
\left\|\tU\tV^\tranH\right\|_{\tcalG,\mathsf{F}}=
\frac{1}{d}\sum_{i=1}^{d}\sum_{k=1}^{n}\frac{1}{w_k}\left|\la \tU_i\tV_i^\tranH, \mG_{k} \ra \right|^2.
\end{align*}
Then the bound for $\left\|\tU\tV^\tranH\right\|_{\tcalG,\mathsf{F}}$ follows immediately from Lemma~\ref{lemma : general} by noting that
\begin{align*}
\max_{k} \frac{1}{d}\sum_{i=1}^{d}\vecnorm{\ve_k^\tran\tU_i\tV_i^\tranH}{2}^2= \max_{k} \frac{1}{d}\sum_{i=1}^{d}\vecnorm{\ve_k^\tran\tU_i}{2}^2 \leq \frac{\mu_0r}{n},
\end{align*}
where the inequality follows from the incoherence assumption. 

The bound for $\left\| \tU\tV^\tranH\right\|_{\tcalG, \infty}$ follows from a direct calculation:
\begin{align*}
 \max_{j,k} \frac{1}{\sqrt{dw_k}} \left| \la \tU\tV^\tranH, \tmG_{j,k} \ra \right|&= \max_{j,k} \frac{1}{\sqrt{dw_k}} \left| \sum_{i=1}^{d}\overline{(\ve_i^\tran\mF\ve_j)}\la\tU_i\tV_i^\tranH,\mG_k \ra \right|\\
	&\leq \max_{j,k} \frac{1}{d} \sum_{i=1}^{d}\left| \frac{1}{\sqrt{w_k}}\la\tU_i\tV_i^\tranH,\mG_k \ra \right|\\
	&\leq \max_{j,k} \frac{1}{d} \sum_{i=1}^{d} \fronorm{\tU_i^\tranH\mG_k}\fronorm{\mG_k\tV_i}\\
	&\leq \max_{j,k} \sqrt{\frac{1}{d} \sum_{i=1}^{d} \fronorm{\tU_i^\tranH\mG_k}^2}\cdot \sqrt{\frac{1}{d} \sum_{i=1}^{d}\fronorm{\mG_k\tV_i}^2}\\
	&\leq \frac{\mu_0r}{n},
	\end{align*}
which completes the proof.
\end{proof}

%\input{conclusion}
%%%%%%%%%%%%%%%%
%\section*{Acknowledgments}
%%%%%%%%%%%%%%%%
\appendix
\section*{Appendix}
%%%%%%%%%%%%%
\section{Proof of Theorem~\ref{main results: special}}
\label{sec:special}
\begin{proof}[Proof of Theorem~\ref{main results: special}]
First note that when $\mX^\natural$ has the special form of \eqref{eq:specialX}, the $2$-th to $n$-th
columns of the solution to \eqref{opt: our nuclear norm} must be zeros.
With a slight abuse of notation, let $\Omega\subset[d]$ be the subset of indices corresponding to the observed entries of $\mX^\natural(:,1)$.  Noting that $\mF\bone=\sqrt{d}\ve_1$, the recovery program \eqref{opt: our nuclear norm} can be reduced to basic pursuit \cite{chen2001atomic}:	
	\begin{align}
	\label{opti: LP}
	\minimize_{\hat{\vx}\in\C^d}~\vecnorm{\hat{\vx}}{1}~\text{subject to}~\mD_{\Omega}\mF^{-1}\hat{\vx} = \mD_{\Omega}\mF^{-1}(\sqrt{d}\ve_1),
	\end{align}
	where $\mD_{\Omega}\in\R^{d\times d}$ is a diagonal matrix of the form
	\begin{align*}
	[\mD_{\Omega}]_{j,k} = \begin{cases}
	1,&j\in\Omega \text{ and }j=k,\\
	0,&j\neq k.
	\end{cases}
	\end{align*}
	
	In order to show that $\sqrt{d}\ve_1$ is the unique optimal solution to \eqref{opti: LP}, by \cite[Lemma 2.1]{candes2006robust} or \cite[Condition 1]{zhang2016one}, we need to 
	%By Theorem in {\color{red}CITE}, the vector $\sqrt{d}\ve_1$ is the unique solution to \eqref{opti: LP} if
	check that the submatrix $(\mD_{\Omega}\mF^{-1})({:,1})$ has full column rank and construct a dual certificate $\lambda\in\R^d$ such that 
	\begin{align}
	\label{eq: dual certificate}
	[\mF\mD_{\Omega}\lambda]_1 = 1,\quad \big| [\mF\mD_{\Omega}\lambda]_j \big| < 1, \text{ for }j=2,\cdots,d.
	\end{align}
	Suppose $\Omega$ has at least two entries, denoted $k_1$ and $k_2$. The full column rank property of the submatrix holds since  $(\mD_{\Omega}\mF^{-1})({:,1}) = \frac{1}{\sqrt{d}}\mD_{\Omega}\bone$ is a vector. %Moreover, such dual certificate exists even when the set $\Omega$ just includes two entries. More precisely, suppose $\Omega = \{ k_1,k_2 \}$ obeying 
	Moreover, assuming 
	\begin{align}
	\label{eq: condition for special certificate}
	(j-1)|k_1-k_2| \neq qd,~ 2\leq j \leq d
	\end{align}
	for any positive integer  $q$, we claim that $\lambda = \frac{\sqrt{d}}{|\Omega|}\bone_{\Omega}$ satisfies the two conditions listed in \eqref{eq: dual certificate}.  The first condition holds since 
	\begin{align*}
	[\mF\mD_{\Omega}\lambda]_1 = \frac{\sqrt{d}}{|\Omega|}[\mF\mD_{\Omega}\bone_{\Omega}]_1 = \frac{\sqrt{d}}{|\Omega|}[\mF\bone_{\Omega}]_1 = 1.
	\end{align*}
	Letting $\Omega^c = \Omega\backslash\{k_1,k_2\}$, for the second condition, a simple algebra yields that 
 \begin{align*}
 \left|\left(\mF\mD_{\Omega}\lambda\right)[j]\right| =& \frac{\sqrt{d}}{|\Omega|}\left|\mF_{j,k_1} +\mF_{j,k_2} + \mF_{j,\Omega^c} \right|\\
 \leq & \frac{\sqrt{d}}{|\Omega|}\left|\mF_{j,k_1} +\mF_{j,k_2}\right| + \frac{\sqrt{d}}{|\Omega|}\sum_{k\in\Omega^c} \left|\mF_{j,k} \right|\\
 \leq& \frac{\sqrt{d}}{|\Omega|}\left|\mF_{j,k_1} +\mF_{j,k_2}\right| + \frac{|\Omega|-2}{|\Omega|}\\
 =&\frac{1}{|\Omega|}\left| \exp\left( -\frac{2\pi i}{d}(j-1)(k_1-1) \right) + \exp\left( -\frac{2\pi i}{d}(j-1)(k_2-1) \right) \right|+\frac{|\Omega|-2}{|\Omega|}\\
 =&\frac{1}{|\Omega|}\left|1+\exp\left(\frac{2\pi i}{d}(j-1)(k_1-k_2)\right)\right|+\frac{|\Omega|-2}{|\Omega|},
 \end{align*}
	which is strictly less than $1$ unless there exists a positive integer $q$ such that $(j-1)|k_1-k_2|=qd$ for all $2\leq j\leq d$. 
	
	It remains to check when \eqref{eq: condition for special certificate} holds. Noting that when $k_1$ and $k_2$ have different parities, $|k_1-k_2|$ is an odd number, so $|k_1-k_2|$ and $d=2^L$ do not have common factors. Therefore for  any positive integer $q>0$, 
		\begin{align*}
	\frac{q\cdot d}{|k_1-k_2|}=\frac{q\cdot 2^L}{|k_1-k_2|} 
	\end{align*}
	is either a fraction or an integer multiple of $d=2^L$. In both cases  \eqref{eq: condition for special certificate} holds since $j-1\leq d-1$.
 Finally, it is not difficult to show that, under the Bernoulli sampling model, $\Omega$ includes at least two indices with different parities with probability at least $\left(1-(1-p)^d - dp(1-p)^{(d-1)}\right)/2$. Clearly, when $d$ is sufficiently large, this value is approximately equal to $0.5$.
\end{proof}
%%%%%%%%%%%%%
\section{Proof of Lemma \ref{lemma incoherence 1}}\label{proof lemma incoherence 1}
\subsection{Proof of \eqref{ineq 1} in Lemma \ref{lemma incoherence 1}}
For any fixed $k$, we have
\begin{align*}
\frac{1}{d}\sum_{a=1}^{d}\fronorm{\mG_k^\tranH\tU_a}^2&= \frac{1}{d}\sum_{a=1}^{d}\fronorm{\frac{1}{\sqrt{w_k}}\sum_{i+j=k}\ve_j\ve_i^\tran\tU_a}^2= \frac{1}{d}\sum_{a=1}^{d}\frac{1}{w_k}\sum_{i+j=k}\vecnorm{\ve_i^\tran\tU_a}{2}^2\\
&= \frac{1}{w_k}\sum_{i+j=k}\frac{1}{d}\sum_{a=1}^{d}\vecnorm{\ve_i^\tran\tU_a}{2}^2\\
&\leq \frac{1}{w_k}\sum_{i+j=k}1\cdot \left(\max_i \frac{1}{d}\sum_{a=1}^{d}\vecnorm{\ve_i^\tran\tU_a}{2}^2\right)\leq \frac{\mu_0r}{n},
\end{align*}
where the last inequality is due to the average case incoherence condition. This completes the proof of the first part of  \eqref{ineq 1} and the second part of \eqref{ineq 1} can be proved similarly.
\subsection{Proof of \eqref{ineq 2} in Lemma \ref{lemma incoherence 1}}
Recalling the definition of $\tT$ and $\tcalG$, for each $(j,k)$, we have
\begin{align*}
\fronorm{\calP_{\tT}\tcalG(\ve_j\ve_k^\tran)}^2 &=\sum_{a=1}^{d}\fronorm{ \calP_{\tT_a}\calG(\ve_a^\tran\mF\ve_j\ve_k^\tran) }^2\\
%&=\sum_{a=1}^{d}\left| \ve_a^\tran\mF\ve_j \right|^2\cdot \fronorm{ \calP_{\tT_a}\calG(\ve_k^\tran) }^2\quad( \ve_a^\tran\mF\ve_j \text{ is a scalar and }\calG,\calP_{\tT_a} \text{ are linear operators})\\
%&=\sum_{a=1}^{d}\left| \ve_a^\tran\mF\ve_j \right|^2\cdot \fronorm{ \calP_{\tT_a}\mG_k }^2\quad( \ve_a^\tran\mF\ve_j \text{ is a scalar and }\calG,\calP_{\tT_a} \text{ are linear operators})\\
&= \frac{1}{d}\sum_{a=1}^{d}\fronorm{ \calP_{\tU_a}(\mG_k) +\calP_{\tV_a}(\mG_k) - \calP_{\tU_a}(\mG_k)\calP_{\tV_a} }^2\\
&= \frac{1}{d}\sum_{a=1}^{d} \fronorm{ \tU_a\tU_a^\tran\mG_k +\mG_k\tV\tV^\tran - \tU_a\tU_a^\tran\mG_k\tV\tV^\tran }^2\\
%&= \frac{1}{d}\sum_{a=1}^{d}\fronorm{ \tU_a\tU_a^\tran\mG_k(\mI-\tV\tV^\tran) +\mG_k\tV\tV^\tran }^2\\
%&= \sum_{a=1}^{d}\left| \ve_a^\tran\mF\ve_j \right|^2\cdot\left( \fronorm{ \tU_a\tU_a^\tran\mG_k(\mI-\tV\tV^\tran)}^2 + \fronorm{\mG_k\tV\tV^\tran }^2 \right)\\
%&= \sum_{a=1}^{d}\left| \ve_a^\tran\mF\ve_j \right|^2\cdot\left( \fronorm{ \tU_a\tU_a^\tran\mG_k}^2 + \fronorm{\mG_k\tV\tV^\tran }^2 \right)\\
&= \frac{1}{d}\sum_{a=1}^{d} \left( \fronorm{ \tU_a\tU_a^\tran\mG_k}^2 + \fronorm{\mG_k\tV\tV^\tran }^2 \right)\\
&\leq \frac{1}{d}\sum_{a=1}^{d} \left( \fronorm{ \tU_a^\tran\mG_k}^2 + \fronorm{\mG_k\tV}^2 \right)\\
&\leq \frac{2\mu_0r}{n}
\end{align*}
%where the last inequality is due to
%\begin{align*}
%\left| \ve_a^\tran\mF\ve_j \right|  = \frac{1}{\sqrt{d}}\left| w^{(a-1)(j-1)} \right| =  \frac{1}{\sqrt{d}}
%\end{align*} 
%for any $1\leq a,j\leq d$. If we define the following incoherence condition, then we have
where the last inequality is due to \eqref{ineq 1}.
%%%%%%%%%%%%%
\section{Auxiliary technical lemmas}
\label{sec Auxiliary}
Here we provide two additional technical lemmas which have been used in the proof of the main result. The proofs of  these two lemmas are  straightforward extensions of those for the $d=1$ case \cite{chen2014robust}. We include the proofs to keep the
presentation self-contained.
\begin{lemma}
	\label{lem:app01}
	Under the condition  $\opnorm{\frac{1}{p}\calP_{\tT}\tcalG\calP_{\Omega}\tcalGT\calP_{\tT} -\calP_{\tT}\tcalG\tcalGT\calP_{\tT} } \leq\frac{1}{2}$, 
	\begin{align*}
	\fronorm{\calP_{\tT}(\tW)} \leq \frac{2\sqrt{2}}{p}\fronorm{\calP_{\tT^\perp}(\tW)}
	\end{align*}
	holds for any $d$-block diagonal matrix $\tW$ which obeys
	\begin{align*}
	\tcalG\calP_{\Omega}\tcalGT(\tW) = \bzero\in\R^{dn_1\times dn_2}, \quad\text{ and }(\calI - \tcalG\tcalGT )(\tW) =  \bzero\in\R^{dn_1\times dn_2}.\numberthis\label{eq:tw_p}
	\end{align*}
\end{lemma}
\begin{proof}
Since $\tW$ satisfies \eqref{eq:tw_p}, we have 
\begin{align*}
0 &= \fronorm{\frac{1}{p}\tcalG\calP_{\Omega}\tcalGT(\tW) + (\calI - \tcalG\tcalGT )(\tW) }\\
&=\fronorm{\left(\frac{1}{p}\tcalG\calP_{\Omega}\tcalGT+(\calI - \tcalG\tcalGT ) \right)\left(\calP_{\tT}(\tW) + \calP_{\tT^\perp}(\tW)\right)}\\
&\geq \fronorm{\left(\frac{1}{p}\tcalG\calP_{\Omega}\tcalGT+(\calI - \tcalG\tcalGT ) \right) \calP_{\tT}(\tW) }
 - \fronorm{\left(\frac{1}{p}\tcalG\calP_{\Omega}\tcalGT+(\calI - \tcalG\tcalGT ) \right)  \calP_{\tT^\perp}(\tW) }.
\end{align*}

On the one hand,
\begin{align*}
	\fronorm{\left(\frac{1}{p}\tcalG\calP_{\Omega}\tcalGT+(\calI - \tcalG\tcalGT ) \right) \calP_{\tT}(\tW) }^2=~& \fronorm{\frac{1}{p}\tcalG\calP_{\Omega}\tcalGT\calP_{\tT}(\tW)}^2 + \fronorm{(\calI - \tcalG\tcalGT )\calP_{\tT}(\tW) }^2 \\
	~&\quad + \frac{2}{p}\left\la\tcalG\calP_{\Omega}\tcalGT\calP_{\tT}(\tW), (\calI - \tcalG\tcalGT )\calP_{\tT}(\tW) \right\ra\\
	=~&\fronorm{\frac{1}{p}\tcalG\calP_{\Omega}\tcalGT\calP_{\tT}(\tW)}^2 + \fronorm{(\calI - \tcalG\tcalGT )\calP_{\tT}(\tW) }^2\\
	=~& \frac{1}{p^2}\left\la\tcalG\calP_{\Omega}\tcalGT\calP_{\tT}(\tW), \tcalG\calP_{\Omega}\tcalGT\calP_{\tT}(\tW) \right\ra  \\
	&\quad +  \left\la(\calI - \tcalG\tcalGT )\calP_{\tT}(\tW), (\calI - \tcalG\tcalGT )\calP_{\tT}(\tW) \right\ra \\
	\geq~& \frac{1}{p}\left\la\calP_{\tT}\tcalG\calP_{\Omega}\tcalGT\calP_{\tT}(\tW), \calP_{\tT}(\tW) \right\ra  +  \left\la\calP_{\tT}(\calI - \tcalG\tcalGT )\calP_{\tT}(\tW), \calP_{\tT}(\tW) \right\ra\\
	=~& \fronorm{ \calP_{\tT}(\tW) }^2 + \left\la\calP_{\tT}(\frac{1}{p}\tcalG\calP_{\Omega}\tcalGT - \tcalG\tcalGT)\calP_{\tT}(\tW), \calP_{\tT}(\tW) \right\ra \\
	\geq~& \fronorm{ \calP_{\tT}(\tW) }^2 - \opnorm{\calP_{\tT}(\frac{1}{p}\tcalG\calP_{\Omega}\tcalGT - \tcalG\tcalGT)\calP_{\tT} }\cdot \fronorm{ \calP_{\tT}(\tW) }^2 \\
	\geq~&  \frac{1}{2}\fronorm{ \calP_{\tT}(\tW) }^2,
	\end{align*}
	where the third line follows from the property $(\tcalG\calP_{\Omega}\tcalGT)(\calI - \tcalG\tcalGT )=0$, and the last line follows from the assumption.

On the other hand, 
\begin{align*}
	 \fronorm{\frac{1}{p}\tcalG\calP_{\Omega}\tcalGT\calP_{\tT^\perp}(\tW) + (\calI - \tcalG\tcalGT )\calP_{\tT^\perp}(\tW) }& \leq \fronorm{\frac{1}{p}\tcalG\calP_{\Omega}\tcalGT\calP_{\tT^\perp}(\tW) } + \fronorm{(\calI - \tcalG\tcalGT )\calP_{\tT^\perp}(\tW) }\\
	&\leq \frac{1}{p}\fronorm{\calP_{\tT^\perp}(\tW) } + \fronorm{\calP_{\tT^\perp}(\tW) }\\
	&\leq \frac{2}{p}\fronorm{\calP_{\tT^\perp}(\tW) },
	\end{align*}
	where the third line follows from the fact that $\tcalG\calP_{\Omega}\tcalGT$ is a projection operator. 
	
	Combining two parts together, we complete the proof.
\end{proof}
%%%%
\begin{lemma}
	\label{lemma : general}
Suppose a $d$-block diagonal matrix $\tZ\in\C^{dn_1\times dn_1}$ satisfies 
\begin{align}
\label{eq: condition 1}
\max_{p} \frac{1}{d}\sum_{i=1}^{d}\vecnorm{\ve_p^\tran\tZ_i}{2}^2 \leq B,
\end{align}
where $\tZ_i\in\C^{n_1\times n_1}$ is the $i$th block of $\tZ$ and $n_1=(n+1)/2$.
Then 
\begin{align*}
\frac{1}{d}\sum_{i=1}^{d}\sum_{k=1}^{n}\frac{1}{w_k}\left|\la \tZ_i, \mG_{k} \ra \right|^2\lesssim B\log(n_1). 
\end{align*}
\end{lemma}
\begin{proof}
Note that 
\begin{align*}
\frac{1}{d}\sum_{i=1}^{d}\sum_{k=1}^{n}\frac{1}{w_k}\left|\la \tZ_i, \mG_{k} \ra \right|^2& = 
\frac{1}{d}\sum_{i=1}^{d}\sum_{k=1}^{n}\frac{1}{w_k}\left|\frac{1}{\sqrt{w_k}}\sum_{p+q=k+1}\la \tZ_i, e_pe_q^\tran \ra \right|^2\\
&\leq \frac{1}{d}\sum_{i=1}^{d}\sum_{k=1}^{n}\frac{1}{w_k}\sum_{p+q=k+1}\left|\la \tZ_i, e_pe_q^\tran \ra \right|^2\\
&=\frac{1}{d}\sum_{i=1}^{d}\sum_{k=1}^{n_1}\frac{1}{w_k}\sum_{p+q=k+1}\left|\la \tZ_i, e_pe_q^\tran \ra \right|^2+\frac{1}{d}\sum_{i=1}^{d}\sum_{k=n_1+1}^{n}\frac{1}{w_k}\sum_{p+q=k+1}\left|\la \tZ_i, e_pe_q^\tran \ra \right|^2.
\end{align*}
We follow the splitting scheme in \cite[Lemma 12]{chen2014robust} to bound  each term. 

Without loss of generality, assume $n_1$ is a power of $2$. We divide $[n_1]$ into $\log(n_1)$ groups:
$
\mathcal{W}_a = \{k: w_k\in[2^{a-1},2^a)\},~a=1,\cdots,\log(n_1).
$ For the first term we have 
\begin{align*}
\frac{1}{d}\sum_{i=1}^{d}\sum_{k=1}^{n_1}\frac{1}{w_k}\sum_{p+q=k+1}\left|\la \tZ_i, e_pe_q^\tran \ra \right|^2 &= \frac{1}{d}\sum_{i=1}^{d}\sum_{a=1}^{\log n_1}\sum_{k\in\calW_a}\frac{1}{w_k}\sum_{p+q=k+1} \left|\la \tZ_i, \ve_p\ve_q^\tran \ra\right|^2\\
&\leq \frac{1}{d}\sum_{i=1}^{d}\sum_{a=1}^{\log n_1}\sum_{k\in\calW_a}\frac{1}{2^{a-1}}\sum_{p+q=k+1} \left|\la \tZ_i, \ve_p\ve_q^\tran \ra\right|^2\\
&\leq \frac{1}{d}\sum_{i=1}^{d}\sum_{a=1}^{\log n_1}\frac{1}{2^{a-1}}\sum_{k=1}^{2^a-1}\sum_{p+q=k+1} \left|\la \tZ_i, \ve_p\ve_q^\tran \ra\right|^2\\
&\leq \frac{1}{d}\sum_{i=1}^{d}\sum_{a=1}^{\log n_1}\frac{1}{2^{a-1}}\sum_{k=1}^{2^a-1}\sum_{p+q=k+1} \left| \ve_p^\tran\tZ_i\ve_q\right|^2\\
&=\frac{1}{d}\sum_{i=1}^{d}\sum_{a=1}^{\log n_1}\frac{1}{2^{a-1}}\sum_{p=1}^{2^a-1}\sum_{q=1}^{2^a-1-p} \left| \ve_p^\tran\tZ_i\ve_q\right|^2\\
&\leq\frac{1}{d}\sum_{i=1}^{d}\sum_{a=1}^{\log n_1}\frac{1}{2^{a-1}}\sum_{p=1}^{2^a-1}\left\|\ve_p^\tran\tZ_i\right\|^2\\
&=\sum_{a=1}^{\log n_1}\frac{1}{2^{a-1}}\sum_{p=1}^{2^a-1}\left(\frac{1}{d}\sum_{i=1}^{d}\left\|\ve_p^\tran\tZ_i\right\|^2\right)\\
&\leq 2B\log(n_1),
\end{align*}
where the second line follows from $w_k\geq 2^{a-1}$ when $k\in \mathcal{W}_a$. 

The second term can be bounded similarly which concludes the proof.
\end{proof}

%%%%%%%%%%%%%%%%
%%%%%%%%%%%%%%%%
\bibliographystyle{unsrt}
\bibliography{refs}
\end{document}